\documentclass[a4paper,onecolumn,10pt,unpublished]{quantumarticle}
\pdfoutput=1
\usepackage[utf8]{inputenc}
\usepackage[english]{babel}
\usepackage[T1]{fontenc}
\usepackage{amsmath}
\usepackage{mathtools}
\usepackage{hyperref}
\usepackage{verbatim}
\usepackage{tikz}
\usepackage{quantikz} 

\usepackage{lipsum}
\usepackage{cite}
\usepackage{amsmath,amsfonts}
\usepackage{amstext} 
\usepackage{array}   
\usepackage{xcolor}
\usepackage{amsthm}
\usepackage{algpseudocode}
\usepackage[english]{babel}
\usepackage[utf8]{inputenc}
\usepackage{braket}
\usepackage{amsfonts}%
\usepackage{amssymb}
\usepackage{subcaption}
\usepackage{tabularx}
\usepackage[linesnumbered,ruled,vlined]{algorithm2e}
\usepackage{soul}
\usepackage{bm}
\usepackage{mdframed}
\usepackage{enumitem}

\makeatletter
\newtheorem*{rep@theorem}{\rep@title}
\newcommand{\newreptheorem}[2]{%
\newenvironment{rep#1}[1]{%
 \def\rep@title{#2 \ref{##1}}%
 \begin{rep@theorem}}%
 {\end{rep@theorem}}}
\makeatother

\newtheorem{lemma}{Lemma}
\newtheorem{theorem}{Theorem}
\newreptheorem{theorem}{Theorem}
\newreptheorem{lemma}{Lemma}
\newtheorem{corollary}{Corollary}

\newtheorem*{remark*}{Remark}

\newcommand{\olive}[1]{\textcolor{olive}{#1}}

\usepackage{pythonhighlight}

\newtheoremstyle{example}{}{}{}{}{\bfseries}{\smallskip}{\newline}{}
\theoremstyle{example}
\newtheorem{example}{Example}

\usepackage{ifpdf}
\usepackage{epstopdf}

\newcommand{\integers}[2]{[#1, #2]}
\newcommand{\RM}[0]{\mathrm{RM}}
\newcommand{\QRM}[0]{\mathrm{QRM}}
\newcommand{\pRM}[0]{\mathrm{pRM}}
\newcommand{\pQRM}[0]{\mathrm{pQRM}}
\newcommand{\zQRM}[0]{\mathrm{QRM}^{(0)}}
\newcommand{\pzQRM}[0]{\pQRM^{(0)}}
\newcommand{\rim}[0]{\texttt{row\_index\_list}}
\newcommand{\Eval}[2]{\mathrm{Eval}^{(#1)}\left(#2\right)}
\newcommand{\wt}[1]{\mathrm{wt}\left(#1\right)}

\newcommand{\CNOT}[2]{\mathrm{CNOT}\left(#1\rightarrow#2\right)}
\DeclareMathOperator*{\argmin}{arg\,min}

\newcommand{\PJ}[1]{\olive{[PJ:#1]}}

\begin{document}

\title{Efficient recursive encoders for quantum Reed-Muller codes towards Fault tolerance}
\author{Praveen~Jayakumar}
\email{praveen.jayakumar@mail.utoronto.ca}
\orcid{0000-0002-1523-8260}
\affiliation{Department of Electronic Systems Engineering, Indian Institute of Science, Bengaluru, India, 560012}
\affiliation{Chemical Physics Theory Group, Department of Chemistry, University of Toronto, Toronto, Canada, M5S 3H6}
\author{Priya~J.~Nadkarni}
\affiliation{Department of Electronic Systems Engineering, Indian Institute of Science, Bengaluru, India, 560012}
\orcid{0000-0002-1351-2959}
\email{priya@alum.iisc.ac.in}
\author{Shayan Srinivasa Garani}
\email{shayangs@iisc.ac.in}
\orcid{0000-0002-2459-1445}
\affiliation{Department of Electronic Systems Engineering, Indian Institute of Science, Bengaluru, India, 560012}
\maketitle
\begin{abstract}
    Transversal gates are logical gate operations on encoded quantum information that are efficient in gate count and depth, and are designed to minimize error propagation. Efficient encoding circuits for quantum codes that admit transversal gates are thus crucial to reduce noise and realize useful quantum computers. The class of punctured Quantum Reed-Muller codes admit transversal gates. We construct resource efficient recursive encoders for the class of quantum codes constructed from Reed-Muller and punctured Reed-Muller codes. These encoders on $n$ qubits have circuit depth of $O(\log n)$ and lower gate counts compared to previous works. The number of CNOT gates in the encoder across bi-partitions of the qubits is found to be equal to the entanglement entropy across these partitions, demonstrating that the encoder is optimal in terms of CNOT gates across these partitions. Finally, connecting these ideas, we explicitly show that entanglement can be extracted from QRM codewords.
\end{abstract}
\section{Introduction}
Noise introduces errors in physical gate realizations and leads to imprecise representations of quantum information on quantum computers. Quantum error correction codes (QECCs) are utilized to store quantum information with redundant information in larger Hilbert spaces, such that the redundancy can be utilized to detect and correct some errors in the state. Since the encoded states are themselves quantum states, we can repeatedly encode the previously encoded states to enable correction of more errors. This process, known as \textit{code concatenation} seems to suggest all errors can be corrected eventually by sufficiently repeating the encoding and we can achieve \textit{fault tolerant} computation\cite{knill1996threshold}. Unfortunately, the encoding process and the logical gate operations used to manipulate the stored quantum information are prone to noise and leads to more errors.

For QECCs to be useful and to achieve fault tolerance, the errors introduced by encoding and error correction gates require to be lower than the errors correctable by the code. For this condition to be satisfied, the noise is required to be below a threshold determined by the code and the complexity of the circuits involved. A larger threshold will reduce the resources required and make fault tolerance with more noisy quantum hardware practical. The threshold can be increased by encoders with lower physical gate counts. Additionally, some QECCs have particularly efficient logical gates called transversal gates that minimize the spread of errors, thus increasing the threshold further. Thus, designing efficient unitary\footnote{States encoded with QECCs are eigen states of operators known as stabilizers. We can prepare these encoded states by measuring the stabilizers\cite{chuang_book_2010}. Such an encoding process may lead to uncorrectable errors in the state, and in certain hardwares, mid-circuit measurements may be not favorable during the initial encoding.} encoding circuits for QECCs that admit transversal gates are crucial to realize fault tolerance. It was shown in Ref. \citenum{eastin_restrictions_2009} that there cannot exists QECCs that admit a universal transversal gate set. Most codes do not admit transversal $H$ and $T$ gates which are necessary to realize a universal set.

Punctured Quantum Reed-Muller (pQRM) codes form a class of QECCs that admit transversal H and T gates\cite{knill1996threshold}. Quantum Reed-Muller (QRM) codes are constructed by importing the classical Reed-Muller (RM) codes into the Calderbank-Shor-Steane (CSS)\cite{calderbank_good_1996} framework and pQRM are likewise obtained by importing punctured RM codes, or simply dropping the first qubit in the corresponding QRM code. We find that, QRM and pQRM codes can be decomposed and expressed in terms of smaller QRM and pQRM codes. Utilizing this property, we construct efficient, recursive encoding circuits for these classes of codes. These circuits are similar to the iterative decoders in the literature of RM codes. In QECCs, redundancy is introduced by encoding message states along with ancillae qubits, usually initialized to the state $\ket{0}$. When the code is concatenated, the ancillae for the second iteration of encoding are encoded versions of $\ket{0}$. We provide recursive state preparation encoders to prepare these ancillary states to encode QRM code states.

Encoding quantum information introduces entanglement into the quantum state to store information. Supposing this entanglement can be extracted from the code states, QECCs provide a reliable error corrected cache of entanglement which can be utilized when required. A recent work by Bravyi et al. have considered this usecase of QRM codes to establish EPR pairs between pairs of parties with local operations and classical communication \cite{Bravyi2024generatingkeprpairs}. In the current work, we show extraction of entanglement amongst all parties, from the QRM codeword resource state with only local operations given by the decoders for QRM codes of smaller length. In addition, we find that the number of CNOT gates across bi-partitions in the recursive encoder constructions are equal to the entanglement across these bi-partitions, as quantified by the measure proposed in Ref. \citenum{fattal_entanglement_2004}. We argue that the proposed encoders are optimal in terms of non-local operations.

This article is organized as follows: In Section \ref{sec:prelim}, we present a brief summary of RM and QRM codes, and action of various quantum gates on a quantum state. We present an efficient encoder achieved by row transforms in Section \ref{sec:rred}. We then present our main results: We present the recursive QRM code construction in Section \ref{sec:decomp}, followed by recursive encoders for the QRM code. We discuss the circuit efficiency, encoder properties, and demonstrate entanglement extraction in Section \ref{sec:disc}. Sufficient illustrative examples are interspersed in the main text to assist the readers.

\section{Preliminaries}\label{sec:prelim}
We present a brief overview of notations, definitions, and constructions of classical and quantum Reed-Muller codes.

\subsection{Classical Reed Muller codes}
A Reed-Muller (RM) code, parameterized by $(r, m)$ is a linear binary code of length $2^m$, code dimension $\sum_{i = 0}^r{m \choose i}$ and distance $2^{m-r}$. It is a $[2^m, \sum_{i = 0}^r{m \choose i}, 2^{m-r}]$ linear code and is denoted by $\RM(r, m)$. There exists multiple equivalent constructions of this code, we present a few here. Detailed derivation and properties are discussed in Appendix \ref{app:RMprop} and can be found in previous works \cite{abbe_reedmuller_2021, Nadkarni2024entanglement}.
\begin{itemize}
    \item Evaluation vector construction: The codewords of $\RM(r, m)$ are given by the evaluation vectors of multivariate polynomials from the polynomial ring $\mathbb{W}_m:= \mathbb{F}_2[x_1, \cdots, x_m]$ of degree $\leq r$. We define the evaluation vector of a polynomial $f$, $\mathrm{Eval}^{(m)}()$ as the concatenation of the values of the function at all points of the domain of the polynomial. 
    Denote the set of positive integers from $i$ to $j$ (inclusive) as:
    \begin{equation}\label{def:positiveintegers}
        \integers{i}{j} := \{i, i+1, \dots, j\}.
    \end{equation}
    For the case of binary codewords defined by $f\in \mathbb{W}_m$, the domain consists of $m$ bit binary vector representation $\bm\rho(j) \in \mathbb F_2^{m}$ of integers $j \in \integers{1}{2^m}$ and the order of points $\bm \rho(j)$ are chosen of increasing $j$ as
    \begin{equation}\label{def:RMeval}
        \mathrm{Eval}^{(m)}(f) := \left(f(\bm \rho{(1)}), \cdots, f(\bm \rho{(2^m)})\right), 
    \end{equation}
    where we define
    \begin{equation}
        \bm \rho{(j)} := (\rho_1(j), \cdots, \rho_m(j)) = \mathrm{bin}(j-1)
    \end{equation}
    as the $m$-bit binary representation of $(j-1)$\footnote{Throughout this work, we index bits, qubits, rows, and columns beginning with $1$.}.
    Since function $\mathrm{Eval}^{(m)}(\cdot)$ is linear (Lemma 1 of \cite{Nadkarni2024entanglement}), the rows of the generator matrix of the $\RM(r, m)$ code is given by the evaluation vectors of the monomials upto degree $r$.
    
    \item Tensor Product construction: We define the Hamming weight $\wt{\bm b}$ of a binary string $\bm b$ as the number of non-zero entries in $\bm b$. The generator matrix of $\RM(r, m)$ is given by the subset of rows with weight $\geq 2^{m-r}$ of the matrix $\left[\begin{array}{cc}
        1 & 1 \\
        0 & 1
    \end{array}\right]^{\otimes m}$.

    \item Plotkin $(u, u+v)$ construction: The codewords of $\bm c = \RM(r, m)$ are constructed recursively as $\bm c = (\bm u, \bm u \oplus \bm v)$ for codewords $\bm u \in \RM(r, m-1), \bm v \in \RM(r-1, m-1)$. Thus the generator matrix is recursively constructed as
    \begin{equation}\label{def:Grm}
        G(r, m) = \begin{bmatrix}
            G(r, m-1) & G(r, m-1) \\ 0 & G(r-1, m-1)
        \end{bmatrix}
    \end{equation}
    with $G(r, r)$ given by $\left[\begin{array}{cc}
        1 & 1 \\
        0 & 1
    \end{array}\right]^{\otimes r}$ (or any set of $2^m$ rows that form a basis of $\mathbb{F}_2^{2^m}$) and $G(0, m) = \underbrace{[1~\cdots~1]}_{2^m} \equiv [\bm 1]$.  Note that $G(-1, m) =\phi$.
\end{itemize}
The dual code is also an RM code,
\begin{equation}\label{prop:RMdual}
    \RM(r, m)^\perp \equiv \RM(m-r-1, m)
\end{equation}
Another useful property of RM codes is the nested codespace structure
\begin{equation}\label{prop:RMsubgroup}
    \RM(-1, m) \subset \RM(0, m) \subset \RM(1, m) \subset \dots \subset \RM(m-1, m) \subset \RM(m, m),
\end{equation}
which follows as an immediate consequence of the evaluation vector construction.

\textit{Generalizing the Plotkin construction}: The evaluation points $\bm \rho{(j)}$, $j\in \integers{1}{2^m}$ in the domain can be thought of as the $2^m$ vertices of a unit hypercube of $m$ dimensions. The split of the bits in the Plotkin construction as $\bm u$ and $\bm u\oplus \bm v$ can be visualized as a $m-1$ dimensional hyperplane $x_1 = 0.5$ splitting the vertices of the $m$ dimensional hypercube into two sets. This corresponds to considering the set of points $\bm \rho{(j)}$ with $\rho_1{(j)} = 0$ in Eq. \eqref{def:RMeval} for the first set and $\rho_1{(j)} = 1$ in Eq. \eqref{def:RMeval} for the second set. The codewords $\bm u$ and $\bm u \oplus \bm v$ over these sets of bit positions is given by evaluation vectors of $f \in \mathbb F_2[x_2, x_3, \cdots, x_m]$ of degree $\leq r$.

Now by splitting the bits of the codeword along the hyperplane $x_i = 0.5$, $i\in \integers{1}{m}$, we generalize the Plotkin construction to \textit{Plotkin-$i$ construction}. This is equivalent to partitioning the bit indices into sets
\begin{equation}
    B^{(m)}(i; k) := \{j|\Eval{m}{x_i}_j = k\},\label{bitsubsets}
\end{equation}
where $B^{(m)}(i; 0)$, $B^{(m)}(i; 1)$ denotes bit positions where the evaluation vectors of $x_i$ are $0$, $1$ respectively. We note that the classical code $\RM(r, m)$ has an Automorphism group $\mathbb{GA}(m)$, whose action on the codewords is defined by affine bijections over the evaluation points  $\bm \rho \in \mathbb F_2^m$ as \cite{geiselhart2021iterative} 
\begin{equation}
    \bm \rho \rightarrow A\rho + \bm b, ~~~ A \in \mathbb F_2^{m\times m}, \bm b\in \mathbb F_2^m.
\end{equation}
The Plotkin-$i$ construction can been seen as an immediate consequence of setting $A$ to be the appropriate permutation matrix, $\bm b = \bm 0$, and applying the Plotkin $(u, u+v)$ construction on the codewords.

\begin{example}\label{ex:rm23}
    Consider the code $\RM(2, 3)$, which is a $[8, 7, 2]$ linear binary code. We can detect errors of Hamming weight at most $d-1 = 1$, but cannot correct any error. We show the construction of the generator matrix.
    
    We start with the tensor product
    \begin{align}
        \left[\begin{array}{p{0.02cm}p{0.02cm}}
             1 & 1 \\ 0 & 1
        \end{array} \right]^{\otimes 3} 
        = \left[\begin{array}{p{0.02cm}p{0.02cm}p{0.02cm}p{0.02cm}p{0.02cm}p{0.02cm}p{0.02cm}p{0.02cm}}
             1 & 1 & 1 & 1 & 1 & 1 & 1 & 1\\
             0 & 1 & 0 & 1 & 0 & 1 & 0 & 1\\
             0 & 0 & 1 & 1 & 0 & 0 & 1 & 1\\
             0 & 0 & 0 & 1 & 0 & 0 & 0 & 1\\
             0 & 0 & 0 & 0 & 1 & 1 & 1 & 1\\
             0 & 0 & 0 & 0 & 0 & 1 & 0 & 1\\
             0 & 0 & 0 & 0 & 0 & 0 & 1 & 1\\
             0 & 0 & 0 & 0 & 0 & 0 & 0 & 1\\
        \end{array} \right].
    \end{align}
    Retaining the rows with weight $\geq 2^{m-r} = 2$ (removing the last row), we obtain the generator matrix of $\RM(2, 3)$. Equivalently, the rows of the generator matrix is given by the evaluation vectors of monomials $f(x_1, x_2, x_3)\in \mathbb{W}_3$, $\mathrm{deg}(f) \leq 2$. We obtain
    \begin{equation}
        G(2, 3) \hspace{-0.1cm}= \hspace{-0.2cm} \left[\hspace{-0.1cm}\begin{array}{p{0.02cm}p{0.02cm}p{0.02cm}p{0.02cm}p{0.02cm}p{0.02cm}p{0.02cm}p{0.02cm}}
             1 & 1 & 1 & 1 & 1 & 1 & 1 & 1\\
             0 & 1 & 0 & 1 & 0 & 1 & 0 & 1\\
             0 & 0 & 1 & 1 & 0 & 0 & 1 & 1\\
             0 & 0 & 0 & 1 & 0 & 0 & 0 & 1\\
             0 & 0 & 0 & 0 & 1 & 1 & 1 & 1\\
             0 & 0 & 0 & 0 & 0 & 1 & 0 & 1\\
             0 & 0 & 0 & 0 & 0 & 0 & 1 & 1\\
        \end{array} \right] \hspace{-0.1cm}\equiv\hspace{-0.1cm} \left[\hspace{-0.1cm}\begin{array}{l}
             \Eval{m}{1} \\
             \Eval{m}{x_3}\\
             \Eval{m}{x_2}\\
             \Eval{m}{x_2x_3}\\
             \Eval{m}{x_1}\\
             \Eval{m}{x_1x_3}\\
             \Eval{m}{x_1x_2}
        \end{array}\hspace{-0.2cm}\right].
    \end{equation}
    The generator matrix is constructed recursively as
    \begin{align}
        G(2, 3) =& \left[\begin{array}{cc}
             G(2, 2) & G(2, 2) \\
             0 & G(1, 2)
        \end{array} \right] \\=& \left[\begin{array}{cc}
             G(2, 2) & G(2, 2) \\
             0 & \left[\begin{array}{cc}
                 G(1, 1) & G(1, 1) \\
                 0 & G(0, 1)
             \end{array}\right]
        \end{array} \right] \\
        =& \left[\begin{array}{p{0.02cm}p{0.02cm}p{0.02cm}p{0.02cm}p{0.02cm}p{0.02cm}p{0.02cm}p{0.02cm}}
             1 & 1 & 1 & 1 & 1 & 1 & 1 & 1\\
             0 & 1 & 0 & 1 & 0 & 1 & 0 & 1\\
             0 & 0 & 1 & 1 & 0 & 0 & 1 & 1\\
             0 & 0 & 0 & 1 & 0 & 0 & 0 & 1\\
             0 & 0 & 0 & 0 & 1 & 1 & 1 & 1\\
             0 & 0 & 0 & 0 & 0 & 1 & 0 & 1\\
             0 & 0 & 0 & 0 & 0 & 0 & 1 & 1\\
        \end{array} \right].
    \end{align}
    For a message vector $\bm m = [1010101]$, we obtain the codeword $\bm c = \bm m G(2, 3) = [11000000]$. If we receive the bit string $\bm c' = [11010000]$, errors can be detected by evaluating the syndrome $s =\bm c'H^T = [1]\neq [0]$ ($H = G(0, 3) = [11111111]$). Since the syndrome has non-zero entries, we conclude that an error has occurred.
\end{example}

\textit{Punctured Reed-Muller code}: Another code of interest is the Punctured Reed-Muller code $\pRM(r, m)$, obtained by puncturing the $\RM(r, m)$ code. Puncturing involves removing a bit position, and in this case the first bit of the code. Denote the puncturing operation by $\cdot^*$ and the punctured version of the codeword $\bm c$ by $\bm c^*$. The punctured Reed-Muller code $\pRM(r, m)$ is a $[2^m -1, \sum_{i=0}^r {m \choose i}, 2^{m-r}-1]$ linear code, requiring that $r \leq m-1$. The Tensor product constructions follow from simply dropping the first column. Extending the Plotkin $(u, u+v)$ construction, the codewords $ \bm c \in \pRM(r, m) $ can be written as $\bm c = (\bm u^*, \bm u \oplus \bm v)$ for some $\bm u \in \RM(r, m-1), \bm v\in \RM(r-1, m-1)$. The generator matrix of $\pRM(r, m)$ is then given by $G(r, m)^*$, where $[\cdot]^*$ denotes the removal of the first column of the matrix. The evaluation vector construction follows from removing $\rho(\bm 1)$ from the points of evaluation. Similar to Eq. \eqref{prop:RMsubgroup}, the punctured codes have a nested codespace structure,
\begin{equation}\label{prop:pRMsubgroup}
    \pRM(-1, m) \subset \pRM(0, m) \subset \pRM(1, m) \subset \cdots \subset \pRM(m-1, m)
\end{equation}
and the dual code is given by
\begin{equation}\label{prop:pRMdual}
    \pRM(r, m)^\perp \equiv \pRM(m-r-1, m) /\{\bm 1\}.
\end{equation}

\subsection{Quantum codes}
Quantum codes can be constructed by identifying commuting observables and mapping the message quantum state to a codeword state in their shared eigenspace. By systematically measuring the eigenvalues of these observables on the received state, errors can be detected and corrected.

\subsubsection{Calderbank Shor Steane codes}
A standard method of constructing a quantum error correcting code (QECC) is the Calderbank Shor Steane (CSS) construction, wherein the QECC is constructed by utilizing two classical error correcting codes $C_1$ and $C_2$. This construction requires that the classical codes satisfy the dual containing criteria 
\begin{equation}\label{CSS:condition}
    C_1^\perp \subseteq C_2,
\end{equation}
equivalently $H_1H_2^T = 0$, where $H_1, H_2$ are the parity check matrices of the classical codes $C_1, C_2$, respectively. This condition ensures that the stabilizer generators form a commuting set. The basis for the quantum codespace is given by
\begin{equation}\label{def:CSS}
    \ket{\bm{w}}_{C_1, C_2} := \frac{1}{\sqrt{|C_1^\perp|}}\sum_{\bm{c} \in C_1^\perp}\ket{\bm w \oplus \bm{c}},
\end{equation}
for classical codewords $\bm{w} = \bm m [G(C_2)\setminus G(C_1^\perp)] \in C_2$, where $G(C_i)$ denotes the generator matrix of $C_i$. This codeword is a superposition of classical states that represent codeword strings in a coset of $C_1^\perp$. Since $\bm x \oplus C_1^\perp = C_1^\perp$ for $\bm x \in C_1^\perp$, for the CSS codewords we have the property
\begin{align}
    X(\bm x)\ket{\bm w}_{C_1, C_2} &= \ket{\bm w\oplus \bm x}_{C_1, C_2}\label{CSS:addx}\\ 
    \Rightarrow X(\bm x)\ket{\bm w}_{C_1, C_2} &= \ket{\bm w}_{C_1, C_2}, ~~ \forall \bm x \in C_1^\perp \label{CSS:Xinv}
\end{align}
where we denote $X(\bm x) = X^{x_1}\otimes\dots \otimes X^{x_{n}}$ as the Pauli product operator consisting of Pauli $X$ operator at qubit positions given by the non-zero entries of $\bm x$. Similarly, 
\begin{equation}\label{CSS:Zinv}
    Z(\bm z)\ket{\bm w}_{C_1, C_2} = \ket{\bm w}_{C_1, C_2}, ~~ \forall \bm z \in C_2^\perp
\end{equation}
where we denote $Z(\bm z) = Z^{z_1}\otimes \dots \otimes Z^{z_{n}}$ as the Pauli product operator consisting of Pauli $Z$ operator at qubit positions given by the non-zero entries of $\bm z$. These operators form a set of stabilizers of the codewords in Eq. \eqref{def:CSS} which form a basis for the $+1$ eigenspace of these operators. We introduce a Pauli-to-binary isomorphism to represent a Pauli operator on $n$ qubits by a binary vector of $2n$ bits where the first $n$ bits represent the $X$ operators, followed by $n$ bits that represent the $Z$ operators such that $(\bm x, \bm z) \equiv i^{\bm x\cdot \bm z}X^{x_1}Z^{z_1}\otimes\cdots\otimes Z^{z_{n}}X^{x_{n}}$. By stacking such binary vectors representing the stabilizers of the CSS codewords, the parity checks of the CSS code $\text{CSS}(C_1, C_2)$ can be represented by the matrix
\begin{equation}\label{paritycheck}
    H = \left[\begin{array}{c|c}
        H_1 & 0 \\
        0 & H2
    \end{array}\right],
\end{equation}
where the first $n$ columns represent the $X$ operators and the last $n$ columns represent the $Z$ operators. Since the full Hilbert space has a dimension of $2^n$ and each independent stabilizer reduces the codespace dimension by a factor of 2, the codespace of the CSS code constructed with classical codes $[n, k_1, d_1]$ and $[n, k_2, d_2]$ will have a dimension of $2^{n-(n-k_1)-(n-k_2)} = 2^{k_1 + k_2 - n}$. The distance of the CSS code is given by at least the minimum weight of the stabilizers, $\min(d_1, d_2)$. We denote the CSS code parameters as $[[n, k_1 + k_2 - n, \geq \min(d_1, d_2)]]$.

\textit{Unitary encoder for CSS codes}: Let $\ket{\phi} \in \mathcal{H}_{2^k}$, be a normalized $k$-qubit state defined over a Hilbert space of dimensions $2^{k}$, where $k = k_1 + k_2 - n$.
We say that $U$ is a unitary encoder of $\mathrm{CSS}(C_1, C_2)$ if $U \in \mathcal U(\mathcal{H}_{2^n})$ is a unitary operator on $n$ qubits such that
\begin{equation}
    U\ket{\phi}\ket{0}^{\otimes n-k} = \sum_{\mathclap{\bm m \in \mathbb F_2^k}} \alpha_{\bm m} \ket{\bm w(\bm m)}_{C_1, C_2}, ~~~ \forall \ket{\phi} \in \mathcal{H}_{2^k}
\end{equation}
where $\bm w(\bm m) = \bm m [G(C_2)\setminus G(C_1^\perp)]$ is a classical encoding of $\bm m$ and $\alpha_{\bm m} = \braket{\bm m | \phi}$. Note that since choice of $G(C_1)$ and $G(C_2)$ is arbitrary upto an invertible row transformation of generators, the unitary encoding circuit of a CSS code is not unique.

\subsubsection{Quantum Reed Muller codes}
We refer to the CSS code constructed with $C_1 = C_2 = \RM(r, m)$ as Quantum Reed Muller code, denoted by $\QRM(r, m)$.  To satisfy Eq. \eqref{CSS:condition}. we require $r\geq (m-1)/2$. The basis codewords of this code, following the construction in Eq. \eqref{def:CSS} for the message basis state $\ket{\bm m}$ is given by
\begin{equation}\label{def:QRM}
    \ket{\bm w}_{r,m} := \frac{1}{\sqrt{|\RM(r, m)^\perp|}}\sum_{\bm c \in \RM(r, m)^\perp} \ket{\bm w \oplus \bm{c}}
\end{equation}
for $\bm w = \bm m[G(r, m)\setminus G(m-r-1, m)]$. 
For $(m-1)/2< r \leq m$,  $\QRM(r, m)$ is a $[[2^m, \sum_{i=m-r}^{r}{m \choose i}, 2^{m-r}]]$ quantum code. When $r = (m-1)/2$, $\QRM(r, m)$ is a zero rate code. Since all correctable errors have distinct syndromes, $\QRM(r, m)$ is a non-degenerate code, see Appendix \ref{app:degeneracy} for proof.

\subsubsection{Zero rate quantum Reed Muller codes}
CSS codes constructed from classical codes $C_1 = \RM(m-r-1, m)$ and $C_2 = C_1^{\perp} = \RM(r, m)$ form a $[[2^m, 0, 2^{\min(r+1, m-r)}]]$ quantum code. We refer to this code as zero rate quantum Reed Muller code, and denoted as $\zQRM(r, m)$. This class of CSS codes have zero rate and have only one codeword state,
\begin{equation}\label{def:zQRM}
    \ket{\bm 0}_{r, m}^{(0)} := \frac{1}{\sqrt{|\RM(r, m)|}}\sum_{\bm c
    \in \RM(r, m)}\ket{\bm c} \equiv \ket{\bm 0}_{m-r-1, m}.
\end{equation}

We define an additional class of QRM codes that are useful later in this work, $\zQRM(r, m; r_{\diamond}, m_{\diamond})$ with codewords of the form
\begin{align}
    \ket{\bm w}_{r, m}^{(0)} =& X(\bm w)\ket{\bm 0}_{r, m}^{(0)},\\
    =& \frac{1}{\sqrt{|\RM(r, m)|}}\sum_{\mathrlap{\bm c
    \in \RM(r, m)}}\ket{\bm w \oplus \bm c},\label{def:zQRMgen}
\end{align}
where $\bm w \in \RM(r_{\diamond}, m)$. Note that the codewords in Eq. \eqref{def:zQRMgen} are not necessarily CSS codewords.

\subsubsection{Punctured quantum Reed Muller codes}
Punctured Reed-Muller codes $\pRM(r, m)$ are obtained by removing the first bit. The dual code is $\pRM(r, m)^\perp = \pRM(m-r-1, m)/\{\bm 1\}$. Punctured Quantum Reed-Muller codes $\pQRM(r, m) \equiv \mathrm{CSS}(\pRM(r, m), \pRM(r, m))$ are obtained by removing the first qubit $q_1$, and the stabilizer generators $X(\bm 1)$, $Z(\bm 1)$ from the stabilizer set of the $\QRM(r, m)$ code, yielding a $[[2^m -1, \sum_{m-r}^r {m \choose i} +1, \geq 2^{m-r}-1]]$ code with basis codewords $\ket{\bm w}_{r, m}^{*}$,
\begin{equation}\label{def:pQRM}
    \ket{\bm w}_{r, m}^* := \frac{1}{\sqrt{|\pRM(r, m)^\perp|}}\sum_{\bm c \in \pRM(r, m)^\perp}\ket{\bm w \oplus \bm c}
\end{equation}
where $\bm w\in \pRM(r, m)$ is a $2^m-1$ bit string. 

\subsubsection{Punctured zero rate quantum Reed Muller codes}
Puncturing the zero rate codes $\zQRM(r, m)$ yields a $[[2^m -1, 1, \geq 2^{\min(r+1, m-r)} - 1]]$ code. We denote this code as $\pzQRM(r, m)$. The codewords of $\pzQRM(r, m)$ are
\begin{equation}\label{def:pzQRM}
    \ket{\bm w}_{r, m}^{(*, 0)} := \frac{1}{\sqrt{|\pRM(r, m)/\{1\}|}}\sum_{{\bm c \in \pRM(r, m)/\{1\}}}\ket{\bm w\oplus \bm c}
\end{equation}
where $\bm w \in \pRM(r, m)$. Unlike the zero rate codes, the punctured zero rate codespace is spanned by two codewords $\ket{\bm 0}_{r, m}^{(*, 0)}$ and $\ket{\bm 1}_{r, m}^{(*, 0)}$ and can encode 1 qubit of information. Of special interest are the codes $\pzQRM(r, 2r+1) \equiv \pQRM(r, 2r+1)$ that allows for transversal CNOT and Hadamard gates, and $\pzQRM(r, m)$, $m\geq 3r+1$ that allows for transversal T and CNOT \cite{knill1996threshold, luo_fault-tolerance_2020}, useful towards fault-tolerant quantum computation.

Additionally, extending the above code, we define $\pzQRM(r, m; r_{\diamond}, m_{\diamond})$ as the punctured versions of $\zQRM(r, m; r_{\diamond}, m_{\diamond})$. Note that since $\pzQRM(r, m) \equiv \pzQRM(r, m; r, m)$, we use the general notation of $\pzQRM(r, m; r_{\diamond}, m_{\diamond})$ henceforth to refer to punctured zero rate QRM codes. The basis codewords have the form of Eq. \eqref{def:pzQRM}, with $\bm w \in \pRM(r_{\diamond}, m)$.

\subsection{Action of quantum gates}\label{sec:gate_action}
Encoders for CSS codes can be constructed with CNOT gates and Hadamard gates\cite{chuang_book_2010}. In later sections, we provide efficient circuit constructions for encoding various QRM codes, without measurements. We now introduce and discuss the action of CNOT, Hadamard gates, and qubit permutations to motivate our encoder constructions.

\subsubsection{Hadamard gate}
The Hadamard gate $H$ is a single qubit gate defined such that
\begin{align}
    H\ket{0} &= \frac{\ket{0} + \ket{1}}{\sqrt{2}},~~ H\ket{1} = \frac{\ket{0} - \ket{1}}{\sqrt{2}}.
\end{align}
and can be represented by the matrix
\begin{align}
    H &\equiv \frac{1}{\sqrt{2}}\begin{pmatrix} 1 & 1 \\ 1 & -1 \end{pmatrix}.
\end{align}
Let $H(K)$ denote the array of Hadamard gates acting on the qubits indexed by elements of $K \subseteq \integers{1}{n}$, then
\begin{equation}
    H(K) = \prod_{i\in K} H_i
\end{equation}
where $H_i$ denotes the Hadamard gate acting on the $i^{\text{th}}$ qubit. In particular, for $K = \integers{1}{n}$ and $\ket{\bm b} = \ket{b_1\cdots b_n}$, we have
\begin{equation}
    H(\integers{1}{n})\ket{\bm b} = \frac{1}{2^{n/2}}\sum_{\bm c \in \mathbb F_2^n}(-1)^{\bm b\cdot \bm c}\ket{\bm c}.
\end{equation}
Note that $H$ is a single qubit gate, thus $K$ may be a set or a sequence of indices.

\subsubsection{CNOT gate}
A CNOT gate applies an $X$ or an $I$ gate on a target qubit if the state of the control qubit is $\ket{0}$ or $\ket{1}$ respectively.
We denote the CNOT gate acting on the $j^{\text{th}}$ qubit controlled on the $i^{\text{th}}$ qubit by 
\begin{equation}\label{def:CNOT}
    \CNOT{i}{j} := \ket{0}\bra{0}_iI_j + \ket{1}\bra{1}_iX_j
\end{equation}
where $I_j$, $X_j$ denote the identity or NOT gate applied on the $j^{\text{th}}$ qubit respectively. We note that the CNOT gate is its own inverse, i.e., we have
\begin{equation}\label{prop:CNOT:refl}
    \CNOT{i}{j}\cdot\CNOT{i}{j} = I.
\end{equation}
We denote an array of $s$ CNOT gates with target qubits indexed by elements of $L$ and control qubits indexed by elements of $K$ by
\begin{equation}\label{def:CNOT:array}
    \CNOT{K}{L} := \prod_{i = 1}^{s} \CNOT{k_i}{l_i}
\end{equation}
where $K = (k_1, k_2, \dots, k_s)$ and $L = (l_1, l_2, \dots, l_s)$, for $k_i, l_i \in \integers{1}{n}$ are ordered sequences.

\noindent\textbf{CNOT as a column operation on a generator matrix}: We consider an $n$ qubit state $\ket{\bm b} = \ket{b_1\dots b_{n}}, \bm b\in \mathbb F_2^n$, then
\begin{align}
    \ket{\bm b} \equiv \ket{\dots b_i\dots b_j\dots}
    \xleftrightarrow{\CNOT{i}{j}} \ket{\dots b_i\dots (b_i\oplus b_j)\dots} \equiv \ket{\bm b M^{(ij)}}\label{CNOTbitaddition}
\end{align}
where 
\begin{equation}\label{def:Mij}
    [M^{(ij)}]_{kl} := \delta_{kl} + \delta_{ik}\delta_{jl},
\end{equation}
is a matrix representation of the action of $\CNOT{i}{j}$ on the bit string $\bm b$ with $\delta_{ij} = 1$ if $i=j$ and $0$ otherwise. $M^{(ij)}$ is simply an $n\times n$ identity matrix in $\mathbb F_2^{n \times n}$ with an additional non-zero entry at the $(i, j)^{\mathrm{th}}$ position. In order to understand the action of the $n\times n$ matrix $M^{(ij)}$ on the generator matrix and the binary vector being encoded, we require to extend the generator matrix to $\mathbb F_2^{n \times n}$ and the message vector to $\mathbb F_2^n$ while preserving the codeword to which it is encoded to.
Let $\bm m' = (\bm m, \bm 0) \in \mathbb F_2^n$ be the message vector $\bm m \in \mathbb F_2^k$ padded with zeros and 
\begin{equation}\label{def:G'}
    G' := \left[\begin{array}{c}
        G\\
     G^s
\end{array}\right]
\end{equation}
be the generator matrix $G \in \mathbb F_2^{k\times n}$ padded with $(n-k)$ linearly independent rows $G^s$ such that $G' \in \mathbb F_2^{n\times n}$ is an $n\times n$ full rank matrix\footnote{\label{fullrank}CNOT gates act as rank-preserving row or column transforms on the generator matrix, and $U_I = I$. Thus, only $U_{G'}$ with full rank matrices $G'$ are achievable from $U_I$ by application of CNOT gates. Since the message vector $\bm m$ is padded with zeroes to $\bm m'$, the codeword encoded by $U_{G'}$ is consistent to that originally encoded with generator matrix $G$.}, given the $k$ rows of $G$ that span the code space $C$ are linearly independent. Let 
\begin{equation}
    \ket{\bm b} = \ket{\bm m \cdot G} = \ket{\bm m' G'} = U_{G'}\ket{\bm m'}
\end{equation}
be an $n$ qubit encoding of $\bm m$, where we say $U_{G'}$ is the \textit{encoding circuit} for $G'$. The action of the CNOT gate on $\ket{\bm b}$ is equivalent to performing a column operation on $G'$ since 
\begin{equation}
    \CNOT{i}{j}\ket{\bm b} = \ket{\bm b M^{(ij)}} = \ket{\bm m' G'M^{(ij)}}.
\end{equation}
and the right action of $M^{(ij)}$ on $G'$ adds column $i$ to column $j$ of $G'$. Thus,
\begin{equation}\label{CNOT:col1}
    \CNOT{i}{j}U_{G'} = U_{G'M^{(ij)}}.
\end{equation}
By multiplying Eq. \eqref{CNOT:col1} by $\CNOT{i}{j}$ on both sides from the left, we obtain
\begin{equation}\label{CNOT:col2}
    U_{G'} = \CNOT{i}{j}U_{G'M^{(ij)}}.
\end{equation}

\noindent\textbf{CNOT as a row operation on a generator matrix}: CNOT is reflective, i.e., $\CNOT{i}{j}\cdot \CNOT{i}{j} = I$. Likewise, the matrix representation in Eq. \eqref{def:Mij} squares to identity, $\left(M^{(ij)}\right)^2 = I$. Then we have
\begin{align}
    \ket{\bm m'\cdot G'} =& \ket{\bm m' M^{(ij)}\cdot M^{(ij)}G'}\\
    =& U_{M^{(ij)}G'} \ket{\bm m' M^{(ij)}}\\
    \Rightarrow U_{G'}\ket{\bm m'} =& U_{M^{(ij)}G'}\CNOT{i}{j}\ket{\bm m'}
\end{align}
The left action of $M^{(ij)}$ on $G'$ adds row $\bm g_j$ to $\bm g_i$. The encoding circuit for $G'$ is a CNOT on the input states followed by the encoding circuit for the row transformed generator matrix $M^{(ij)}G'$
\begin{equation}\label{CNOT:row}
    U_{G'} = U_{M^{(ij)}G'}\CNOT{i}{j}
\end{equation}

\subsubsection{Qubit permutations}
Consider the bijective maps $p, p^{-1}:\integers{1}{n} \rightarrow \integers{1}{n}$ such that $p^{-1}(p(i)) = p(p^{-1}(i)) = i ~\forall i \in \integers{1}{n}$. 
$p$ induces a permutation on the qubits $(q_1, \dots, q_{n})$ such that the qubit at position $i$ is permuted to position $p(i)$. The map $p$ can be faithfully represented by the sequence 
\begin{equation}\label{def:permutationstring}
    \overline p := (p(1), \dots, p(n)).
\end{equation}
Define the action of this permutation on the computational basis state $\ket{\bm b}$ as
\begin{align}
    P(p)\ket{\bm b} :=& \ket{b_{p^{-1}(1)}, \dots, b_{p^{-1}(n)}} \equiv \ket{\bm b P(p)}
\end{align}
where
\begin{align}
    [P(p)]_{kl} = \sum_{i = 1}^{n} \delta_{ik}\delta_{p(i)l} = \sum_{i=1}^{n}\delta_{p^{-1}(i)k}\delta_{il}
\end{align}
is the matrix representation of the permutation induced by $p$ on the bit string $\bm b$. 
The permutation matrices have the following properties
\begin{align}
    (i)~& P(p)P(q)= P(q\circ p)\\
    (ii)~& P^T(p) = P^{-1}(p)
\end{align}
where $q\circ p(i) = q(p(i))$ denotes the function composition. Similar to Eqs. \eqref{CNOT:col2}, \eqref{CNOT:row} we can show
\begin{align}
    U_{G'} =& P(p^{-1})U_{G'P(p)}\label{perm:col}\\
    U_{G'} =& U_{P(p)G'}P(p^{-1})\label{perm:row}
\end{align}
By Eq. \eqref{perm:col}, we may apply the encoding circuit for the column permuted matrix $G'P(p)$ followed by a qubit permutation given by $p^{-1}$ to construct the encoding circuit for $G'$. Similarly, by Eq. \eqref{perm:row}, we may encode $G'$ by applying a permutation given by $p^{-1}$ followed by the encoder for row permuted matrix $P(p)G'$.

\noindent\textbf{Qubit permutations and CNOT gates}: Since permutations are simply relabelling of qubits $q_i\rightarrow q_{p(i)}$, we have the permutation
\begin{align}
    P^{-1}(p)&\CNOT{i}{j}P(p) \nonumber\\&= \CNOT{p^{-1}(i)}{p^{-1}(j)}
\end{align}

No operations are required to encode the generator matrix $G' = I_{n\times n}$, i.e., $U_{I} =I$. Thus, using Eqs. \eqref{CNOT:col2}, \eqref{CNOT:row}, \eqref{perm:col}, and \eqref{perm:row} we can perform column and row operations to reduce the generator matrix $G'$ to the identity matrix and reconstruct $U_{G'}$. Using the concepts discussed in this section, we next construct the encoding circuits for QRM codes.

\section{Row reduced encoder for Quantum Reed Muller codes}\label{sec:rred}
Let $n = 2^m$ and $k = \sum_{i=0}^r {m \choose i}$. We consider the initial computational basis state to be encoded $\ket{\bm m}\ket{\bm 0}$ where $\bm m \in \mathbb{F}_2^{2k-n}$ is the message bit string and $\ket{\bm 0}$ are the $2(n-k)$ ancilla qubits. Apply the Hadamard gates $H(\integers{2k-n + 1}{k})$ to the first $n-k$ ancilla qubits to obtain
\begin{equation}\label{RedEnc:hadamard}
    \frac{1}{2^{(n - k)/2}}\sum_{\mathclap{\bm u \in \mathbb F_2^{n-k}}}\ket{\bm m}\ket{\bm u}\ket{\bm 0}
\end{equation}
Let $G_1 = G(r, m)\setminus G(m-r-1, m)$, and $G_2 = G(m-r-1, m)$ denote generator matrices with rows that span $\RM(r, m)/\RM(m-r-1, m)$ and $\RM(m-r-1, m)$, respectively, and define $G^{(s)}$ as a matrix consisting of linearly independent rows such that
\begin{equation}\label{def:Gs}
    G := \left[\begin{array}{c}
          G_1 \\ G_2 \\ G^{(s)}
    \end{array}\right]
\end{equation}
forms a full rank matrix (See Footnote \ref{fullrank}). By Lemma \ref{lem:uniquefirst} in Appendix \ref{app:RMprop}, the column index of the leading non-zero entry of rows in $G_1$ and $G_2$ are unique. Call this set of unique column indices as $L$. We can choose $G^{(s)}$ to be unit vectors, each of length $n$ with a single non-zero entry at positions $[2^m]\setminus L$. Since $(\bm m, \bm u, \bm 0)G = \bm mG_1 \oplus \bm uG_2 \oplus \bm 0$, on applying $U_{G}$ as defined in Section \ref{sec:gate_action} on Eq. \eqref{RedEnc:hadamard}, we obtain the state
\begin{equation}
    U_{G}\frac{1}{2^{n - k - 1}}\sum_{\mathclap{\bm u \in \mathbb F_2^{n-k}}}\ket{\bm m}\ket{\bm u}\ket{\bm 0} = \frac{1}{2^{n-k-1}}\sum_{\mathclap{\bm c \in \RM(r, m)^\perp}}\ket{\bm w \oplus \bm c}
\end{equation}
where $\bm w = \bm m G_1 \in \RM(r, m)$, $\bm c = \bm u G_2 \in \RM(m-r-1, m)$.

Now, we require to construct $U_G$. Since $G$ is a full rank matrix, we can determine a sequence of row/column operations and qubit permutations that transform $G$ into the $n\times n$ identity matrix $I$, and $U_{I}$ is simply the identity operation on $n$ qubits. Reversing this sequence of matrix operations will then determine the sequence of gate operations to construct $U_{G}$. One such sequence is as follows: First, we apply the row permutation $P(p)$ such that $P(p)G$ is in row-echelon form\footnote{$P(p)$ depends on the explicit ordering of rows in G, thus is left to readers to determine as required.}. Next, we use Eq. \eqref{CNOT:col2} to perform column operations to remove every off-diagonal non-zero entry in $P(p)G$ individually starting from the first row. As each off-diagonal entry of $P(p)G$ will require $1$ column operation to be removed, we will require $\sum_{i =0}^r {m \choose i}(2^{m-i} - 1)$ CNOT gates. 

Row and columns operations do not commute, thus the sequence of row and column operations are not unique and can be optimized to reduce the number of gates required. Additionally, to reduce the number of gates, we can choose $G_1, G_2$ with rows that span the same codespaces with lower non-zero entries. This can be achieved by considering row transforms $R_1, R_2$ on $G_1, G_2$ respectively. Note that we do not need extra CNOT gates to reflect this row transformation since this is an alternative choice for the generators. The matrix $G$ in Eq. \eqref{def:Gs} is modified to $G^{(\mathrm{rred})} = RG$, where
\begin{equation}
    R = \left[\begin{array}{ccc}
         R_1 &  0 & 0\\
         0 & R_2 & 0\\
         0 & 0 & I
    \end{array}\right].
\end{equation}
The encoder on $P(p^{-1})\ket{\bm m}\ket{\bm 0}$ is then constructed with Hadamard gates, followed by $U_{P(p)G^{(\mathrm{rred})}}$. We call these encoders \textit{row reduced encoders}. We now provide an explicit construction of a row transform $R$ that optimally reduces the generator matrix.

\noindent\textbf{Constructing row transformation R}:\\
We introduce the notation $x_A = \prod_{i \in A} x_i$ as the monomial product consisting of variables $x_i$ of indices $i\in A, A\subseteq \integers{1}{m}$, where recall that $\integers{i}{j} = \{i, i+1, \dots , j\}$.
Consider a set $A \subseteq \integers{1}{m}, |A| \leq r$ such that $\Eval{m}{x_A} = \bm g$. Define
\begin{equation}\label{rowaddform}
    x_A' = x_A\prod_{i\in S\setminus A}(1 + x_i) = \sum_{B, A\subseteq B\subseteq S}x_B
\end{equation}
where the set $S$ is chosen such that $A \subseteq S\subseteq \integers{1}{m}, |S| = r$. By Lemma \ref{lem:evalwt} (iii) in Appendix \ref{app:RMprop}, $\wt{\Eval{m}{x_A'}} = \Eval{m}{x_S} = 2^{m-r}$. Hence, by adding rows consisting of $\Eval{m}{x_B}, A\subset B\subseteq S$ to $\bm g$, we reduce the number of non-zero entries in $\bm g$ to $2^{m-r}$.

We have the freedom to choose the sets $S$ for every $A$. We choose sets $S$ that minimizes the sum $\sum_{s\in S} s$. Denote this set $S_r(A)$, defined as
\begin{equation}
    S_r(A) = \argmin_{\substack{S\\A\subseteq S\subseteq \integers{1}{m}\\|S|=r}} \sum_{s \in S}s
\end{equation}
For the generator matrix $G(r, m) = \left[\bm g_1^T, \dots, \bm g_k^T\right]^T$, the row transform $R$ that reduces the row weights to $2^{m-r}$ is then given by
\begin{equation}
    [R]_{i, j} = \left\{\begin{array}{cc}
        1 & B \subseteq S_r(A) \\
        0 & \text{otherwise}
    \end{array} \right.
\end{equation}
for sets $A, B \subseteq \integers{1}{m}$ defined such that $\Eval{m}{x_A} = \bm g_i, \Eval{m}{x_B} = \bm g_j$.

Recall that $R_1$ acts on $G_1 \equiv G(r, m)\setminus G(m-r-1, m)$ and $R_2$ acts on the rows $G_2 \equiv G(m-r-1, m)$. Since the rows of $G_1$ have generators of the form $\Eval{m}{x_A}$ for $|A|\leq r$, the evaluation vector of a linear combination of the form in Eq. \eqref{rowaddform} can reduce the row weight to $2^{m-r}$. Likewise, rows of $G_2$ have generators of the form $\Eval{m}{x_A}$ for $|A| \leq m-r-1$ and can be reduced to $2^{r+1}$.

Summarizing, the row transforms are given by
\begin{equation}
    [R_1]_{i, j} = \left\{\begin{array}{cc}
        1 & B \subseteq S_r(A) \\
        0 & \text{otherwise}
    \end{array} \right. ~~ \text{where } A, B : \Eval{m}{x_A} = \bm g_i^{(1)}, \Eval{m}{x_B} = \bm g_j^{(1)}
\end{equation}

\begin{equation}
    [R_2]_{k, l} = \left\{\begin{array}{cc}
        1 & D \subseteq S_{m-r-1}(C) \\
        0 & \text{otherwise}
    \end{array} \right. ~~\text{where } C, D : \Eval{m}{x_C} = \bm g_k^{(2)}, \Eval{m}{x_D} = \bm g_l^{(2)}
\end{equation}
where rows of $G_1$ are indexed as $\bm g_i^{(1)}$ and rows of $G_2$ are indexed as $\bm g_i^{(2)}$.

\textit{CNOT gate counts}:
$R_1G_1$ and $R_2G_2$ have $\sum_{i = m-r}^{r}{m\choose i}$ and $\sum_{i = 0}^{m-r-1}{m \choose i}$ rows respectively. For $m-r-1<r<m$, each row in $R_1G_1$ and $R_2G_2$ have $2^{m-r}$ and $2^{r+1}$ non zero entries respectively, and when $r = m-r-1$, the generator matrix $G_1 = \emptyset$. The non-zero entries in $R_1 G$ and $R_2 G$ are optimal since the codes defined by $G_1$ and $G_2$ have code distances of $2^{m-r}$ and $2^{r+1}$ respectively. The number of CNOT gates $\zeta^{(\mathrm{rred})}(r, m)$ required is then
\begin{align}
    \zeta^{(\mathrm{rred})}(r, m) = \begin{cases}\sum_{i = m-r}^{r}{m\choose i}(2^{m-r}-1) + \sum_{i = 0}^{m-r-1}{m \choose i}(2^{r+1} - 1) & \text{if } m>r>m-r-1\\
    \sum_{i = 0}^{m-r-1}{m \choose i}(2^{r+1} - 1) & \text{if } r = m -r-1\\
    0 & \text{if } r = m
    \end{cases}
\end{align}

\section{Recursive construction of Quantum Reed Muller codewords}\label{sec:decomp}
Using the Plotkin constructions, we show that codewords of QRM codes can be recursively written as a superposition of tensor products of shorter length QRM codewords on subsets of qubits. Using this observation, we then construct recursively constructed encoders for QRM codewords. Similar observation of converting a smaller code to a larger code and vice versa is discussed in Ref. \cite{anderson_fault-tolerant_2014}. Here we do not convert codewords and we consider the larger codeword in its entirety.

\begin{theorem}[Recursive construction of QRM codeword]\label{thm:decomp2}
    For $\lceil \frac{m-1}{2}\rceil \leq r < m$, the codewords of $\hbox{QRM}(r, m)$ on $2^m$ qubits provided in Eq. \eqref{def:QRM} can be written as a superposition of tensor product of codewords of $\hbox{QRM}(r, m-1)$ on the first $2^{m-1}$ qubits and the last $2^{m-1}$ qubits.
    
    For $\bm w = (\bm w_1, \bm w_2) \in \RM(r, m)$ with $\bm w_1,\bm w_2 \in \RM(r, m-1)$,
    \begin{align}
        \ket{\bm{w}}_{r, m} = N &\sum_{ \bm{u} \in B} \ket{\bm{w_1} \oplus \bm{u}}_{r, m-1}
        \ket{\bm{w_2} \oplus \bm{u}}_{r, m-1} \label{decomp}
    \end{align}
    where $B = \RM(m-r-1, m-1)/\RM(m-r-2, m-1)$ is a quotient set and $N = \frac{1}{\sqrt{|B|}}$ is the normalization factor.
\end{theorem}
\begin{proof}
    In the QRM codewords in Eq. \eqref{def:QRM}, we rewrite the codewords $\bm c$ and $\bm w$ in the summation with the Plotkin $(u, u+v)$ construction as 
    \begin{align}
    N_1\sum_{\mathclap{\bm c \in \RM(r, m)^\perp}}\ket{\bm w \oplus \bm c}
    &= N_1\sum_{\bm \alpha \in C}\sum_{\bm \beta \in D}\ket{\bm w_1\oplus\bm \alpha}\ket{\bm w_2\oplus\bm \alpha \oplus \bm \beta}\label{decompstep0}\\
    &= N_1\sum_{\bm \alpha \in C}\left(\ket{\bm w_1\oplus\bm \alpha}\sum_{\bm \beta \in D}\ket{\bm w_2\oplus\bm \alpha \oplus \bm \beta}\right)\label{decompstep1}\\
    &= N_2\sum_{\bm \alpha \in C}\ket{\bm w_1\oplus\bm \alpha}\ket{\bm w_2\oplus\bm \alpha}_{r, m-1}\label{decompstep2}
    \end{align}
    where the summations are over $\bm c =(\bm \alpha, \bm \alpha \oplus \bm \beta)\in \RM(m-r-1, m), \bm \alpha \in C := \RM(m-r-1, m-1), \bm \beta \in D:= \RM(m-r-2, m-1)$ and $\bm w_1, \bm w_2 \in \RM(r, m-1)$ ($\bm w \equiv (\bm w_1, \bm w_1 \oplus \bm w_2') \equiv (\bm w_1, \bm w_2)$ for $\bm w_2' \in \RM(r-1, m-1)$ where we rewrite $\bm w_1 \oplus \bm w_2' = \bm w_2$ for simplicity). $N_1 = 1\Big/\sqrt{|\RM(r, m)^\perp|}$ and $N_2 = 1\Big/\sqrt{|\RM(r-1, m-1)^\perp|}$ are the normalization constants.

    By Property \eqref{prop:RMsubgroup}, we have $\RM(m-r-2, m-1) \subset \RM(m-r-1, m-1)$ by which we can express $\alpha \in \RM(m-r-1, m-1)$ in Eq. \eqref{decompstep2} as $\bm \alpha = \bm u \oplus \bm \beta$ for $\bm u \in B := \RM(m-r-1, m-1)/\RM(m-r-2, m-1)$ and $\bm \beta \in \RM(m-r-2, m-1)$. Equation \eqref{decompstep2} is rewritten as
    \begin{align}
        N_1\sum_{\bm c}\ket{\bm w \oplus \bm c} 
        &= N_2\sum_{\bm u\in B}\sum_{\mathrlap{\!\bm \beta \in \RM(r, m-1)^\perp}}\ket{\bm w_1\oplus\bm u \oplus \bm \beta}\ket{\bm w_2\oplus\bm u \oplus \bm \beta}_{r, m-1}\label{decomp:step:1}\\
        &= N_2\sum_{\bm u\in B}\sum_{\mathrlap{\!\bm \beta \in \RM(r, m-1)^\perp}}\ket{\bm w_1\oplus\bm u \oplus \bm \beta}\ket{\bm w_2\oplus\bm u}_{r, m-1}\label{decomp:step:2}\\
        &= N\sum_{\bm u\in B}\ket{\bm w_1\oplus\bm u}_{r, m-1}\ket{\bm w_2\oplus\bm u}_{r, m-1}\label{decomp:step:3}
    \end{align}
    where $N = 1\Big/ \sqrt{B} = N_2^2/N_1$. Equation \eqref{decomp:step:2} is obtained by Properties \eqref{CSS:addx} and \eqref{CSS:Xinv}, and Eq. \eqref{decomp:step:3} is obtained using the definition of the QRM codewords. The bound $\lceil (m-1)/2\rceil \leq r$ is required to ensure the codeword $\ket{\bm w}_{r, m}$ is a valid CSS codeword. For $r = m$, $\ket{\bm w}_{m, m} = \ket{\bm w}$ is a computational basis state with $\bm w\in \RM(m, m)$. While it can be written as a tensor product of states of $\QRM(m-1, m-1)$, it is no longer a superposition of computational basis states and we restrict the statement of the proof to $r < m$.
\end{proof}

By similar arguments, we can show a recursive structure for punctured and zero rate QRM codewords. See Theorems \ref{thm:decompp}, \ref{thm:decompz}, \ref{thm:decomppz} in Appendix \ref{app:decomposition}.

\textit{Generalizing Theorem \ref{thm:decomp2} with Plotkin-$i$ construction:}Suppose the $2^m$ qubits are labelled as $q_1, q_2, \dots, q_{2^m}$, the Plotkin$-1$ (standard) construction induces the partition consisting of qubit subsets $\{q_1, \cdots, q_{2^{m-1}}\},\{q_{2^{m-1}+1}, \cdots, q_{2^{m}}\}$. Theorem \ref{thm:decomp2} showed that the states over these subsets of qubits are ensembles of codewords of $\QRM(r, m-1)$. Similarly, Plotkin-$i$ construction induces the qubit subsets 
\begin{align}
Q^{(m)}(i; k) =& \{q_j|\Eval{m}{x_i}_j = k\},\label{qubitsubsets}
\end{align}
and Theorem \ref{thm:decomp2} holds true for the qubit subsets $\{Q^{(m)}(i; 0), Q^{(m)}(i; 1)\}$, $\forall i \in \integers{1}{m}$.

\begin{corollary}\label{cor:qubitsetdecomp}
    For $\lceil \frac{m-1}{2}\rceil < r < m-1$, the codewords of $\QRM(r, m)$ on $2^m$ qubits provided in \eqref{def:QRM} can be written as a superposition of tensor product of codewords of $\QRM(r, m-1)$ on the qubit sets given by $\left\{Q^{(m)}(i;0),  Q^{(m)}(i;1)\right\} \forall i \in \integers{1}{m}$.
\end{corollary}

Additionally, since the states over the qubit subsets $Q^{(m)}(1; k)$ for $k=0, 1$ is an ensemble of codewords of $\QRM(r, m-1)$, we can successively apply Corollary \ref{cor:qubitsetdecomp} to each subsystem.
\begin{corollary}\label{cor:decomps}
    For $l \leq m - r$ and $r > \lceil\frac{m-1}{2}\rceil$, the codewords of $\QRM(r, m)$ provided in Eq. \eqref{def:QRM} can be written as a superposition of tensor product of $2^l$ codewords of $\QRM(r, m-l)$.
\end{corollary}

To describe the qubit subsets over which the states remain codewords by Corollary \ref{cor:decomps}, we extend the qubit set defined in Eq. \eqref{qubitsubsets} as
\begin{equation}\label{def:qubitpartitionsets}
    Q^{(m)}(\bm s; \bm t) = \bigcap_{i=1}^l Q^{(m)}(s_i; t_i)
\end{equation}
For instance, consider  $Q^{(4)}\left((1, 2); (1, 0)\right)$. We have qubit sets 
\begin{align}
    Q^{(4)}\left(1; 1\right) =& \{q_{9}, q_{10}, q_{11}, q_{12}, q_{13}, q_{14}, q_{15}, q_{16}\},\nonumber\\
    Q^{(4)}\left(2; 0\right) =& \{q_{1}, q_{2}, q_{3}, q_{4}, q_{9}, q_{10}, q_{11}, q_{12}\},\nonumber
\end{align}
as $\Eval{4}{x_1} = [0000000011111111], \Eval{4}{x_2} = [0000111100001111]$. Then
$$Q^{(4)}\left((1, 2); (1, 0)\right) = Q^{(4)}\left(1; 1\right) \cap Q^{(4)}\left(2; 0\right) = \{q_{9}, q_{10}, q_{11}, q_{12}\}.$$
Consider the state $\ket{0}_{2, 4}\in \QRM(2, 4)$ in Eq. \eqref{QRM24_decomp}. By Corollary \ref{cor:decomps}, the state is a superposition of tensor product of $4$ codewords of $\QRM(2, 2)$. The state over qubits $Q^{(4)}\left((1, 2); (1, 0)\right)$ on tracing out the remaining qubits is an ensemble consisting of states of $\QRM(2, 2)$.

\subsection{Recursive encoders for Quantum Reed-Muller codes}
Using Theorem \ref{thm:decomp2}, we present encoders for $\QRM(r, m)$ constructed recursively with encoders of $\QRM(r, m-1)$. 
On applying $U(r, m-1)^\dagger \otimes U(r, m-1)^\dagger$ on the state in Eq. \eqref{decomp}, we obtain
\begin{equation*}
    N\sum_{\bm u \in B}\ket{\bm m_{\bm w_1 \oplus \bm u}, \bm 0}\ket{\bm m_{\bm w_2 \oplus \bm u}, \bm 0} \equiv N\sum_{\bm u \in B} \ket{\bm m_{\bm w_1} \oplus \bm m_{\bm u}, \bm 0}\ket{\bm m_{\bm w_1} \oplus \bm m_{\bm w_2'} \oplus \bm m_{\bm u}, \bm 0}
\end{equation*}
where we have used the linearity of the code, and $\bm w_2 = \bm w_1 \oplus \bm w_2'$ for $\bm w_2' \in \RM(r-1, m-1)$. Applying CNOT gates from qubits with state $\ket{\bm m_{\bm w_2} \oplus \bm m_{\bm u}}$ to qubits with state $\ket{\bm m_{\bm w_1} \oplus \bm m_{\bm w_2'} \oplus \bm m_{\bm u}}$, we reduce the state to $$N\sum_{\bm u \in B} \ket{\bm m_{\bm w_1} \oplus \bm m_{\bm u}, \bm 0}\ket{\bm m_{\bm w_2'}, \bm 0} \equiv N\sum_{\bm u \in B}P\ket{\bm m_{\bm w_1},\bm m_{\bm w_2'}, \bm m_{\bm u}, \bm 0}$$
where since $\bm w_1 \in \RM(r, m-1)/\RM(m-r-1, m-1)$ and $\bm u\in \RM(m-r-1, m-1)/\RM(m-r-2, m-1)$ span disjoint spaces, we can write $\bm m_{\bm w_1} \oplus \bm m_{\bm u}$ as some permutation $P$ of $(\bm m_{\bm w_1},\bm m_{\bm u})$. We use Hadamard gates to create $N \sum_{\bm u \in B}\ket{\bm m_{\bm u}}$ from $\ket{\bm 0}$. Inverting the permutation $P$ and Hadamard gates ($H^\dagger = H$), we obtain the initial message state $\ket{\bm m_{\bm w_1}, \bm m_{\bm w_2'}, \bm 0}$. By reversing this process, we recursively construct $U(r, m)$ from CNOTs, qubit permutations, and $U(r, m-1)$. The overarching idea behind these recursive encoders is to prepare the message states $$\ket{\bm m_{\bm w_1 \oplus \bm u}, \bm 0}\ket{\bm m_{\bm w_2 \oplus \bm u}, \bm 0}$$ by inverting the steps described above, and in subsequent iterations, we encode with $U(r, m-1)\otimes U(r, m-1)$. 

In order to determine the permutations and qubits to apply CNOT gates on, we create qubit dictionary maps that denote the qubit index positions the encoder expects message qubits. CNOT gates are then used to reversibly copy the input message qubits to the target locations as specified by the maps of the smaller encoders, which are then encoded individually in the recursive iterations. For instance, consider $G = [\bm g_1^T, \dots, \bm g_k^T]^T$. We define $\texttt{row\_index\_map}(G)$ of G such that
\begin{equation}\label{def:rowindexmap}
    \rim(G)[\bm g_i] = i \forall g_i \in G
\end{equation}
where $i$ is the row index of $\bm g_i$. Let $M = \texttt{row\_index\_map}(G)$. Then, notice that $\bm m G = \sum_{\bm g \in G} m_{M[\bm g]}\bm g$. We say the message bit at position ${M[\bm g]}$ ($m_{M[\bm g]}$) \textit{encodes} the generator $\bm g$, i.e., the generator $\bm g$ is added to $\bm w$ conditional on the value of $m_{M[\bm g]}$. 

Now, suppose we wish to encode the state $\ket{\bm c}$ for $\bm c = m_u(\bm u, \bm u) \in G, \bm c \in \mathbb F_2^{2n}$ where $m_u \in \mathbb F_2$ is the bit that encodes the generator $(\bm u, \bm u)$ for $\bm u \in G_1$. Assume we begin with the state $\ket{0 \dots 0 m_u 0\dots 0}$ with $\ket{m_u}$ at qubit $q_{M[\bm c]}$. We are provided with the encoder $U_{G_1}$ and index map $M_1 = \texttt{row\_index\_map}(G_1)$ for the smaller code $G_1$. First, we apply a qubit permutation $P(p)$ such that $p(M(\bm c)) = M_1(\bm u)$. Then, we apply $\CNOT{M_1(\bm u)}{n + M_1(\bm u)}$ to obtain $\ket{\dots 0 m_u 0\dots}\ket{\dots 0 m_u 0\dots}$. This prepares the message state for encoding each subsystem of $n$ qubits with $U_{G_1}$. We now apply $U_{G_1}^{\otimes 2}$ to obtain the state $\ket{m_u\bm u}\ket{m_u\bm u} \equiv \ket{\bm c}$.
\begin{figure}
    \centering
    \begin{subfigure}{\textwidth}
        \centering
        \begin{mdframed}
        \centering{\noindent\textbf{\texttt{RecursiveQRM(r, m)}}, for $\lceil(m-1)/2\rceil\leq r\leq m$\\
        Returns: Encoder $U(r, m)$, qubit index map $M$.}\\
        \begin{flushleft}
            For $r = m$: Return \texttt{RecursiveBasisQRM(m, m)}\\
            For $\lceil(m-1)/2\rceil\leq r<m$:
        \end{flushleft}
        \begin{description}[itemsep=0.5pt]
            \item[1.] Obtain $U(r, m-1)$, $M_1$ from \texttt{RecursiveQRM(r, m-1)}.
            \item[2.] Initialize circuit over $2^m$ qubits. Create $M = \rim(G)$ where
            $$G = \left[\begin{array}{c} G_1 \\ \hline G_2 \vert G_2\end{array}\right], ~~~\text{for } \begin{array}{c} G_1 = G(r, m)\setminus G(m-r-1, m) \\ G_2 = G(m-r-1, m-1)\setminus G(m-r-2, m-1)\end{array}$$
            \item[3.] Add gates $H(K), K = (k+1, \dots, k + {m-1\choose m-r-1})$ for $k = \sum_{i=m-r}^r{m\choose i}$.
            \item[4.] Add qubit permutation $P(p)$ such that $$p(M[(\bm u, \bm u)]) = M_1[\bm u] ~~\forall \bm u \in G_3,$$ $$p(M[(\bm 0, \bm v)]) = 2^{m-1} + M_1[\bm v] ~~\forall \bm v\in G_4$$
            where $G_3 = G(r, m-1)\setminus G(m-r-2, m-1)$ and $G_4 = G(r-1, m-1)\setminus G(m-r-2, m-1)$.
            \item[5.] Add $\CNOT{i}{j}$ for $i = M_1[\bm u]$, $j = 2^{m-1} + i$ $\forall \bm u \in G_3$.
            \item[6.] Add $U(r, m-1)^{\otimes 2}$ to the circuit.
            \item[7.] Return circuit, $M$.
        \end{description}
        \end{mdframed}
        \caption{}
        \label{alg:rec_qrm}
    \end{subfigure}

    \begin{subfigure}{\textwidth}
        \begin{mdframed}
            \centering{\textbf{\texttt{RecursiveBasisQRM(r, m)}}\\
            Returns Encoder $U^c(r, m)$, qubit index map $M$.}\\
        \begin{minipage}[t]{0.45\textwidth}
            \centering{\underline{Case 1:} For $0\leq r < m$:}
            \begin{description}
                \item[1.] Obtain $U^c(r, m-1)$, $M_1$ from $\texttt{RecursiveBasisQRM(r, m-1)}$.
                \item[2.] Initialize circuit over $2^m$ qubits. Create $M = \rim(G(r, m))$.
                \item[3.] Add qubit permutation $P(p)$ such that 
                \begin{align}
                    p\left(M[(\bm u, \bm u)]\right) &= M_1[\bm u]\nonumber\\
                    p\left(M[(\bm 0, \bm v)]\right) &= 2^{m-1} + M_1[\bm v]\nonumber
                \end{align}
                $$ \forall \bm u \in G(r, m-1), \bm v\in G(r-1, m-1).$$        
                \item[4.] Add $\CNOT{i}{j}$ for $i = M_1[\bm u]$, $j = 2^{m-1} + i \forall \bm u \in G(r, m-1)$.
                \item[5.] Add $U^c(r, m-1)^{\otimes 2}$.
                \item[6.] Return circuit, $M$.
            \end{description}
        \end{minipage}
        \hfill
        \begin{minipage}[t]{0.5\textwidth}
            \centering{\underline{Case 2:} For $r = m $:}
            \begin{description}
                \item[1.] Obtain $U^c(m-1, m-1)$, $M_1$ from $\texttt{RecursiveBasisQRM(m-1, m-1)}$.
                \item[2.] Initialize circuit over $2^m$ qubits. Create $M = \rim(G(m, m))$.
                \item[3.] Add qubit permutation $P(p)$ such that 
                \begin{align}
                    p(M[(\bm u, \bm u)]) &= M_1[\bm u]\nonumber\\
                    p(M[(\bm 0, \bm v)]) &= 2^{m-1} + M_1[\bm v]\nonumber
                \end{align}
                $$\forall \bm u \in G(m-1, m-1), \bm v\in G(m-1, m-1)$$       
                \item[4.] Add $\CNOT{i}{j}$ for $i = M_1[\bm u]$, $j = 2^{m-1} + i \forall \bm u \in G(m-1, m-1)$.
                \item[5.] Add $U^c(m-1, m-1)^{\otimes 2}$.
                \item[6.] Return circuit, $M$.
            \end{description}
        \end{minipage}
        \end{mdframed}
        \caption{}
        \label{alg:rec_qrm_c}
    \end{subfigure}
    \caption{(a) Recursive construction for $U(r, m)$ and (b) Recursive construction for $U^c(r, m)$.}
    \label{alg:rec_qrm_all}
\end{figure}

$\texttt{RecursiveQRM(r, m)}$ described in Figure \ref{alg:rec_qrm} recursively constructs the encoder $U(r, m)$. Step $1$ is the recursive call. Step $2$ initializes a qubit index map $M$ that provides information about message qubit locations, which is determined by the rows of $G$. The rows of $G_1, G_2$ are the generators of $C_2 / C_1^\perp$ and $C_1^\perp$ as defined in the CSS construction of QRM code in Eq. \eqref{def:QRM}. Note that for $r = (m-1)/2$, $G_1 = \emptyset$. Step $3$ adds Hadamard gates on ${m-1 \choose m-r-1}$ ancilla qubits, which are then used as message qubits for encoding $\bm u$ in Eq. \eqref{decomp}. Steps $4$ and $5$ prepare the message states for the subsequent recursive iteration by permutation and CNOT gates. Step $6$ completes the construction by adding encoders over the two halves.%

\textit{CNOT Gate count}: Since in Step $5$ we add $\sum_{i =m-r-1}^r {m-1 \choose i}$ CNOT gates and in Step $6$, we add encoders $U(r, m-1)$, the number of CNOT gates $\mathcal \zeta({r, m})$ in the construction of $U(r, m)$ is given by
\begin{equation}
    \mathcal \zeta({r, m}) =
        \sum_{i=m-r-1}^r {m-1 \choose i} + 2\mathcal \zeta(r, m-1)
\end{equation}
for $r \geq m - r - 1$, where we denote the CNOT gates for $U(r, m)$ by $\zeta(r, m-1)$.
Circuits in Figure \ref{fig:RecQRM} illustrates this construction.

\begin{figure}
    \centering
    \begin{subfigure}[b]{0.5\textwidth}
        \centering
        \includegraphics[width=\textwidth]{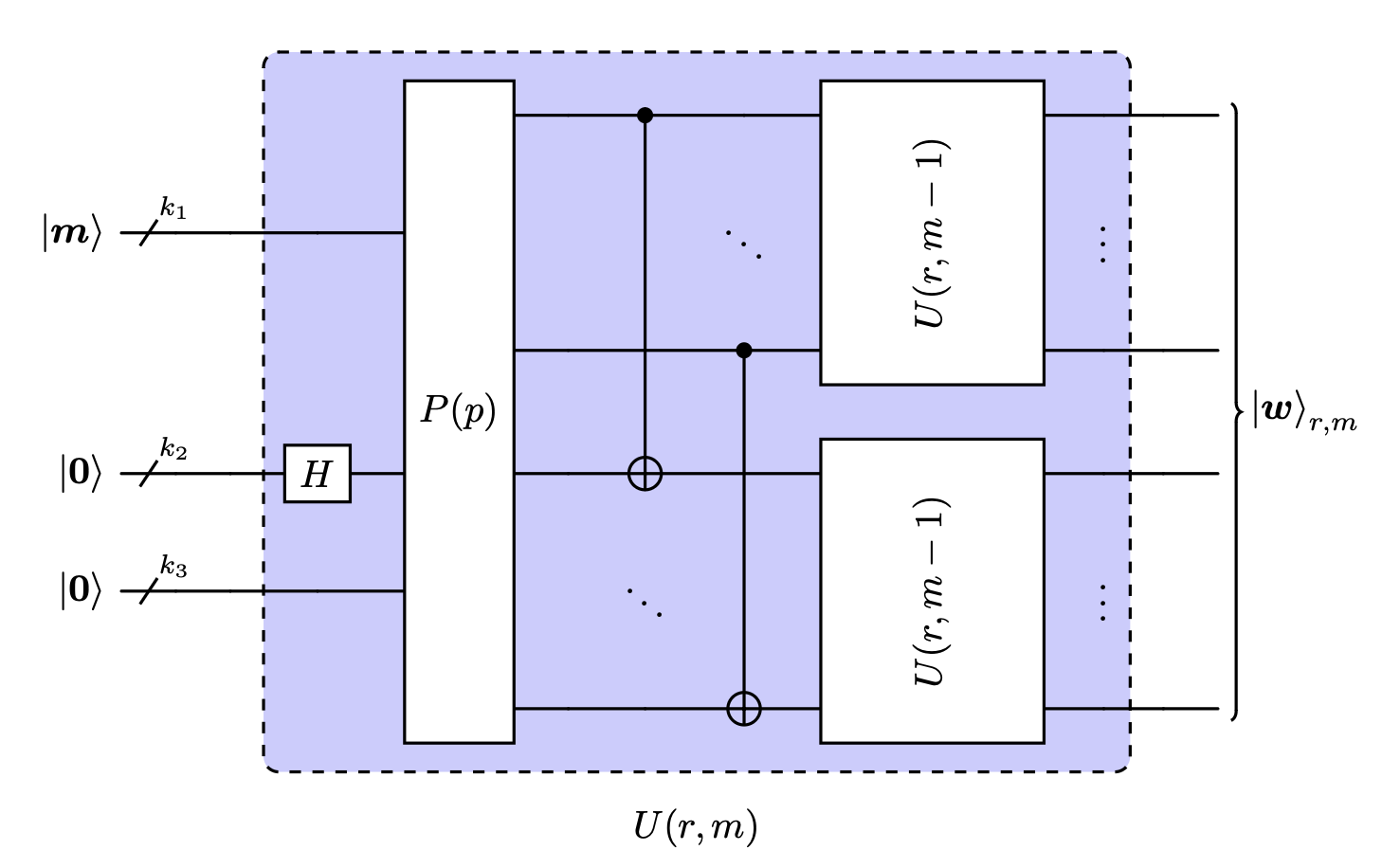}
        \caption{}
        \label{fig:RecQRM}
    \end{subfigure}
    \hfill
    \begin{subfigure}[b]{0.45\textwidth}
        \centering
        \includegraphics[width=\textwidth]{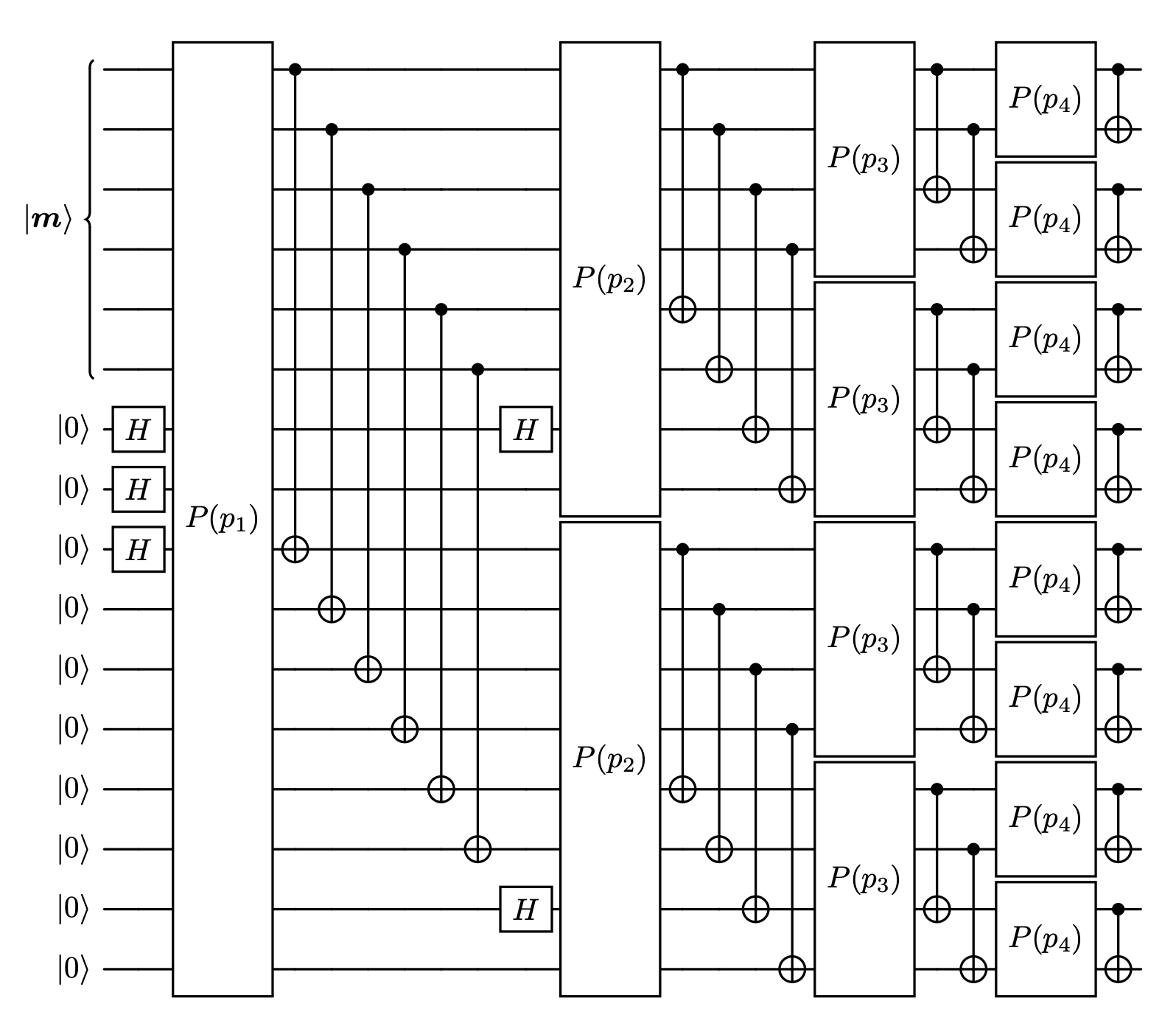}
        \caption{}
        \label{fig:u24}
    \end{subfigure}
    \caption{(a) Recursive encoder for $\QRM(r, m)$, with $k_1 = \sum_{i=m-r}^r {m \choose i}, k_2 = {m-1 \choose m-r-1}$, and $k_3 = 2^m - k_1 - k_2$, (b) Construction of $U(2, 4)$ with permutation $p_1, p_2, p_3, p_4$ as defined in \texttt{RecursiveQRM(r, m)}.}
    \label{fig:RecursiveQRM}
\end{figure}

Note that after $(m-r)$ successive recursive iterations of $\texttt{RecursiveQRM(r, m)}$, the state to be encoded is of the code $\QRM(r, r)$. These states are computational basis states of the form $\ket{\bm m G(r, r)}$, where $\bm m \in \mathbb F_2^{2^r}$ is a message string.  The encoder for these states are distinct to the encoder obtained from the algorithm $\texttt{RecursiveQRM(r, m)}$. We now provide a recursive encoder construction to encode the computational basis states.

\noindent\textbf{Recursive encoder for computational basis states}: We define computational basis encoders as $U^c(r, m) := U_{G(r, m)}$. $\texttt{RecursiveBasisQRM(r, m)}$ in Figure \ref{alg:rec_qrm_c} provides the recursive construction for $U^c(r, m)$, with a constructure similar to $\texttt{RecursiveQRM(r, m)}$.

We begin with the state $\ket{\bm m_{\bm u}, \bm m_{\bm v}, 0}$. In Step 1 and 2, we initialize the circuits and obtain the qubit maps and $U^c(r, m-1)$ for the recursive iteration. In Step 3, we permute the message qubits consisting $\bm m_{\bm u}, \bm m_{\bm v}$ to the positions dictated by $U^c(r, m)$. In Step 4, we add CNOT gates to form the state $\ket{\bm m_{\bm u}, \bm 0}\ket{\bm m_{\bm u \oplus \bm v}, \bm 0}$. In Step 5, we add $U^c(r, m)^{\otimes 2}$ completing the recursive step.

In Example \ref{ex:U24}, we illustrate the construction of $U(2, 4)$ for encoding states of $\QRM(2, 4)$ by the recursive algorithm \texttt{RecursiveQRM}. The full construction of $U(2, 4)$ is illustrated in Figure \ref{fig:u24}.
\begin{example}[Construction of $U(2, 4)$]\label{ex:U24}
    $\QRM(2, 4)$ is a $[[16, 6, 4]]$ code. Let the state to be encoded be a computational basis state $\ket{\bm m}, \bm m \in \mathbb F_2^{6}$ with $10$ ancilla qubits initialized to $\ket{0}$ as
    \begin{equation}\label{u24message}
        \ket{m_1m_2m_3m_4m_5m_6}\ket{0}^{\otimes 10}.
    \end{equation}
    The generator matrix $G$ and the qubit map $M$ in Step 2 of $\texttt{RecursiveQRM}$ are
    \begin{align}
        G = \left[\begin{array}{c} G_1 \\ \hline G_2 \vert G_2\end{array}\right] =  \left[\begin {array}{c}
        \begin{array}{cccc}
                0001 & 0001 & 0001 & 0001 \\
                0000 & 0101 & 0000 & 0101 \\
                0000 & 0011 & 0000 & 0011 \\
                0000 & 0000 & 0101 & 0101 \\
                0000 & 0000 & 0011 & 0011 \\
                0000 & 0000 & 0000 & 1111 \\
                \end{array}\\\hline
                \begin{array}{cc|cc}
                0101 & 0101 & 0101 & 0101 \\
                0011 & 0011 & 0011 & 0011 \\
                0000 & 1111 & 0000 & 1111
        \end{array}
        \end{array}\right] 
        \equiv \left[ \begin{array}{c}
            g_1 \\ g_2 \\ g_3 \\ g_4 \\ g_5 \\ g_6 \\ g_7 \\ g_8 \\ g_9
        \end{array} \right]; ~~~ M[g_i] = i \label{u24generators}
    \end{align}
    for $G_1 = G(r, m)\setminus G(m-r-1, m), G_2 = G(m-r-1, m-1)\setminus G(m-r-2, m-1)$. Likewise, the generator matrix $G'$ in $\texttt{RecursiveQRM(2, 3)}$ is given by
    \begin{align}
        G' = \left[\begin{array}{cc}
                0101 & 0101 \\
                0011 & 0011 \\
                0001 & 0001 \\
                0000 & 1111 \\
                0000 & 0101 \\
                0000 & 0011 \\
                1111 & 1111 
        \end{array}\right] 
        \equiv \left[ \begin{array}{c}
            g_1' \\ g_2' \\ g_3' \\ g_4' \\ g_5' \\ g_6' \\ g_7'
        \end{array} \right]; ~~~ M'[g_i'] = i \label{u23generators}
    \end{align}
    We add Hadamard gates $H$ on qubits $\{q_7, q_8, q_9\}$, Eq. \eqref{u24message} is converted to 
    \begin{equation}\label{u24afterH}
        \frac{1}{2\sqrt{2}} \sum_{u_1, u_2, u_3 = 0}^1 \ket{m_1m_2m_3m_4m_5m_6}\ket{u_1u_2u_3}\ket{0}^{\otimes 7}
    \end{equation}
    Using Eqns. \eqref{u24generators} and \eqref{u23generators}, we define the qubit permutation $P(p)$ in Step 4 given by the string
    \begin{equation}
        \bar p = (3, 5, 6, 9, 10, 12, 1, 2, 4, 7, 8, 11, 13, 14, 15, 16)
    \end{equation}
    where $\bar p$ defines the permutation $P(p)$ as described in Eq. \eqref{def:permutationstring}. Applying the qubit permutation, we transform the state to
    \begin{equation}
        \frac{1}{2\sqrt{2}} \sum_{u_1, u_2, u_3 = 0}^1 \ket{u_1, u_2, m_1, u_3, m_2, m_3, 0, 0, m_4, m_5, 0, m_6, 0, 0, 0, 0}
    \end{equation}
    We now add CNOT gates $\CNOT{i}{i+8}$ for $i = 1, 2,\dots, 6$ to obtain
    \begin{equation}\label{U24penultimatestep}
        \frac{1}{2\sqrt{2}} \sum_{u_1, u_2, u_3 = 0}^1 \ket{u_1, u_2, m_1, u_3, m_2, m_3, 0, 0}\ket{(m_4\oplus u_1), (m_5\oplus u_2), m_1, (m_6\oplus u_3), m_2, m_3, 0, 0}
    \end{equation}
    Finally, we add the encoders $U(2, 3)^{\otimes 2}$ onto the state in Eq. \eqref{U24penultimatestep} to obtain the encoded state.
\end{example}
\textit{CNOT Gate counts}: In Step 4 of the recursive construction, we use $|G(r, m)| = \sum_{i=0}^r {m-1 \choose i }$ CNOT gates.  Thus, $U^c(r, m)$ is constructed with
\begin{align}
    \zeta^c(r, m) = \begin{cases}
        \sum_{i=0}^r{m-1\choose i} + 2\zeta^c(r, m-1) & \text{if } 0\leq r < m, \\
        m2^{m-1} & \text{if } r = m,\\
    \end{cases}
\end{align}
CNOT gates as $\zeta^c(m, m) = 2^{m-1} + 2\zeta(m-1, m-1)$ and $\zeta^c(1, 1) = 1$. Since $U^c(m, m) \equiv U(m, m)$, we have $$\zeta(m, m) = \zeta^c(m, m).$$
The CNOT gate counts for a few code parameters $(r, m)$ is listed in Table \ref{tbl:qrmencoder}.

\textit{Circuit depth}: Since the CNOT gates added at each recursive iteration are on distinct qubits (layer depth $1$) and there are $O(\log n)$ recursive iterations (in the construction of $U(r, m)$ and $U^c(r, m)$), the total circuit depth is $O(\log n)$.

Similar recursive encoder constructions and their gate requirements, for all the classes of punctured and zero rate QRM codes are presented in Appendices \ref{app:decomposition}.

\section{Discussions}\label{sec:disc}
\subsection{Circuit efficiency}
\textit{Entanglement and CNOT gates}:
Since the states considered are codewords of stabilizer codes, we consider the measure of entanglement for stabilizer states introduced by Fattal et al. in Ref \citenum{fattal_entanglement_2004}. Suppose $n$ qubits $A$ are partitioned into $k$ disjoint subsets $\{A_1, A_2, \dots, A_k\}$ such that $A = \bigcup_{i}A_i, A_i\cap A_j = \phi$ for $i\neq j$, then the measure of entanglement is given by
\begin{equation}\label{def:entropy}
    \mathcal E = n - \left|\prod_{j=1}^k S_j\right|
\end{equation}
where $S_j$ are the \textit{local} stabilizers that act as identity on the $j^{\text{th}}$ partition $A_j$ of the qubits and $|S|$ denotes the rank of $S$. $\mathcal E$ is an entanglement monotone as the rank of $S_{loc} = \prod_{j=1}^k S_j$ is invariant under local unitary transformations on every $A_j$, and increases on measurement of local Pauli operators. Stabilizers $S$ of rank equal to the number of qubits is sufficient to completely specify the state. Larger the rank of $S_j$, the more deterministic, and uncorrelated the state over $A_j$ is from the rest.
The authors of Ref. \citenum{fattal_entanglement_2004} also note that for a bi-partition, the entropy of entanglement is simply $e_{A-B} = \mathcal E/2$.
For the bi-partition on qubits $Q^{(m)}(1; 0) = \{q_1, \dots, q_{2^{m-1}}\}$ and $Q^{(m)}(1; 1) = \{q_{2^{m-1} + 1}, \dots, q_{2^m}\}$ of the state $\ket{\bm w}_{r, m}\in \QRM(r, m)$, we have
\begin{equation}\label{def:entropy_bipartition}
    e_{A-B} = \sum_{i = m-r-1}^r {m-1 \choose i}
\end{equation}
as $S_1$, $S_2$ are the stabilizers of $\QRM(r, m-1)$ over qubit sets $Q^{(m)}(1;0)$, $Q^{(m)}(1;1)$ respectively and $S_1\cap S_2 =\phi$. $e_{A-B}$ is the maximum possible entropy of entanglement across this bi-partition. Since a CNOT gate can increase the entanglement entropy $e_{A-B}$ by at most $1$, a valid encoder for $\ket{\bm w}_{r, m}$ will require at least $e_{A-B}$ CNOT gates across this bi-partition. $e_{A-B}$ coincides with the number of CNOT gates across this partition in the construction of the recursive circuit $U(r, m)$. Thus, the recursive construction is optimal in terms of CNOT gates across the bi-partitions at each stage. The CNOT gate counts for the row-reduced and recursive encoders are listed in Table \ref{tbl:qrmencoder}. The recursive encoder has a lower gate count for most code parameters, except when $r\approx m$. By identifying row transformations on the generator matrices of the encoders in the final layer of the recursive construction (when $m = r$ or $m = r+1$ in the case of QRM and pQRM codes respectively) that reduce non-zero entries while keeping the previously added layers of CNOT gates invariant, we can potentially further reduce gate counts towards complete optimality\PJ{To add!!!}. We leave this to future investigation. The CNOT gate counts for a few classes of codes relevant for fault tolerance are listed in Table \ref{tbl:punctured}. Apart from $\ket{+}_{1, 3}^{(*, 0)}$, CNOT counts are found to be lower than previous works.

\textit{Circuit depth}: The encoding circuits for QRM and pQRM codes have depth of $O(\log n)$. The depth for a few state preparation circuits are listed in Table \ref{tbl:punctured}.

\begin{table}[]
    \centering
    \begin{subtable}{0.45\linewidth}
    \centering
    \begin{tabular}{|c|cc|cc|}
    \hline
     & \multicolumn{2}{c|}{CNOT counts}& \multicolumn{2}{c|}{$E_d$}\\\hline
        $(r, m)$ & \multicolumn{1}{c|}{Red.} & Rec. & \multicolumn{1}{c|}{Red.} & Rec. \\\hline
        $(1, 3)$ & 12 & 10 & 3.37 & 3.5\\
        $(2, 4)$ & 53 & 30 & 7.06 & 7.0\\
        $(3, 4)$ & 29 & 32 & 4.5 & 8.12\\
        $(3, 6)$ & 470 & 176 & 16.73 & 15.87\\
        $(4, 6)$ & 367 & 190 & 14.53 & 19.75\\
        $(5, 6)$ & 125 & 192 & 5.84 & 20.78\\
        $(3, 7)$ & 960 & 372 & 20.35 & 18.375\\
        $(4, 7)$ & 1389 & 430 & 26.0 & 26.93\\\hline
    \end{tabular}
    \caption{}
    \label{tbl:qrmencoder}
  \end{subtable}
  \hfill
  \begin{subtable}{0.5\linewidth}
    \centering
    \begin{tabular}{|c|c|cc|c|}
    \hline
         & & \multicolumn{2}{c|}{CNOT counts} & Depth\\\hline
         (r, m) & N &\multicolumn{1}{c|}{\cite{luo_fault-tolerance_2020}} & This work & This work\\\hline
        $(1, 3)$ & $7$ & $8, 8$ & $8, 9$ & $4, 5$\\
        $(1, 4)$ & $15$ & $22, 24$ & $22, 24$ & $5, 6$\\
        $(2, 5)$ & $31$ & $190, 190$ & $63, 65$ & $6, 10$\\
        $(2, 6)$ & $63$ & $248, 258$ & $150, 153$ & $7, 11$\\
        $(2, 7)$ & $127$ & $868, 762$ & $332, 336$ & $8, 12$\\
        \hline
    \end{tabular}
    \caption{Table 2}
    \label{tbl:punctured}
  \end{subtable}
  \caption{(a) CNOT gate count and average error propagation distance, $E_d$ of the row reduced (Red.) and recursive (Rec.) encoders for encoding states of $\QRM(r, m)$. (b) State preparation resources and circuit depth for $QRM^{(*, 0)}(r, m)$ encoders from previous works\cite{luo_fault-tolerance_2020} and recursive encoders presented in this work. CNOT counts for state preparation of $\left(\ket{0}_{r, m}^{(*, 0)}, \ket{+}_{r, m}^{(*, 0)}\right)$ are listed. Exact expressions for gate counts are listed in Appendix \ref{app:punc_zero}. Depth includes the Hadamard gates in the beginning.}
  \label{tbl:gate_counts}
\end{table}

\textit{Error propagation}: There are two main factors that influence occurrence and propagation of errors, the number of gates and the gate connectivity. For the recursive encoders introduced here, while the number of gates introduced are lower, the qubit connectivity is not necessarily lower than previous methods and can lead to spread of errors. We define a simple heuristic metric, the average error propagation distance $E_d$ as the average connectivity of the qubits. We say a qubit $q_1$ is connected to $q_0$ if they are either directly connected by a CNOT gate control (target) or indirectly by a CNOT gate control (target) to another qubit $q_2$ that is directly or indirectly connected to $q_0$ by a CNOT control (target) earlier in the circuit. The connectivity of a qubit is the number of qubits connected to it, and $E_d$ is the average of connectivity over all qubits. Thus, if an $X$ or $Z$ error occurs at any qubit in the beginning of the circuit, on average it would propagate to $E_d$ qubits, ignoring potential cancellations. $E_d$ for a few encoder constructions are provided in Table \ref{tbl:qrmencoder}. We observe that for $r\approx m/2$, the gate counts and $E_d$ are lower for the recursive encoders.

\textit{Hardware locality}: The recursive constructions provide an spatially local encoder. Implementing spatially non-local CNOT gates on qubits with connectivity restricted to a $2-$dimensional plane requires a series of qubit swaps to bring interacting qubits together. Suppose the $2^m$ qubits are partitioned into $\left\{Q^{(m)}(1; 0), Q^{(m)}(1; 1)\right\}$. In the first stage of the construction of the recursive encoders, we apply CNOT gates across this partition, with at most $1$ CNOT gate on each qubit. In the subsequent stages, qubits interactions are restricted to the subsets and do not require to be swapped across. This reduces the instances of qubit swaps required. Figure \ref{fig:staged} illustrates the partition and locality.

\subsection{Distilling entanglement}
By applying decoders on qubit subsets in Eq. \eqref{def:qubitpartitionsets} on the encoded states $\ket{\bm w}_{r, m}$, we obtain entangled pairs of qubits across the partitions. For instance, applying the decoders $U(r, m-1)$ on qubit sets $Q^{(m)}(1;0), Q^{(m)}(1; 1)$ of the state $\ket{\bm 0}_{r, m}, r < m-1$, we obtain ${m-1\choose m-r-1}$ Bell pairs across this partition. In general, we can obtain at least ${m-l\choose m-r-1}$ $2^l$-party $GHZ$ states by considering partitions consisting of subsets of $2^{m-l}$ qubits. This is particularly advantageous since errors can be corrected before decoding and making use of the entanglement. In Example \ref{ex:4ghz}, we show extraction of ${2\choose 1} = 2$ GHZ states across 4 parties. We note that in Ref. \citenum{Bravyi2024generatingkeprpairs}, the authors showed that the code states of $\QRM(r, m)$ can be utilized to distribute EPR (Bell) pairs amongst $k\approx \mathrm{O}(m)$ pairs of parties with local operations and classical communication (LOCCs), whereas in our work, we show the existence of GHZ states shared amongst all parties.

\begin{example}[$4-$party GHZ from $\QRM(2, 4)$]\label{ex:4ghz}
    We begin with the state $\ket{\bm 0}_{2, 4}$. Using Theorem \ref{thm:decomp2} twice, we obtain
    \begin{align}
        \ket{0}_{2, 4} =& \frac{1}{2\sqrt{2}}\sum_{\bm u \in A/B} \ket{\bm u}_{2, 3}\ket{\bm u}_{2, 3}\\
        =& \frac{1}{4\sqrt{2}}\sum_{\substack{\bm \alpha \in C/D\\ \bm \beta \in D}}\left(\sum_{\bm v \in D}\ket{\bm \alpha\oplus \bm v}_{2, 2}\ket{\bm \alpha\oplus \bm \beta\oplus \bm v}_{2, 2}\right)\left(\sum_{\bm w \in D}\ket{\bm \alpha\oplus \bm w}_{2, 2}\ket{\bm \alpha\oplus \bm \beta\oplus \bm w}_{2, 2}\right)\label{QRM24_decomp}
    \end{align}
    where $A = \RM(1, 3), B = \RM(0, 3), C = \RM(1, 2), D = \RM(0, 2)$. On applying the decoders $\left(U(2, 2)^\dagger\right)^{\otimes 4}$ we obtain the state
    \begin{equation}
        \frac{1}{4\sqrt{2}} \sum_{\substack{m_1=0\\ m_2=0\\m_3=0}}^{1} \left[\sum_{m_4 = 0}^1 \ket{m_4, \overset{\downarrow}{m_1}, \underset{\uparrow}{m_2}, 0}\ket{m_4 \oplus m_3, \overset{\downarrow}{m_1}, \underset{\uparrow}{m_2}, 0} \right] \left[\sum_{m_5 = 0}^1 \ket{m_5, \overset{\downarrow}{m_1}, \underset{\uparrow}{m_2}, 0}\ket{m_5 \oplus m_3, \overset{\downarrow}{m_1}, \underset{\uparrow}{m_2}, 0} \right]
    \end{equation}
    Notice that qubits $\{q_2, q_6, q_{10}, q_{14}\}$ and $\{q_3, q_7, q_{11}, q_{15}\}$ (indicated by $\overset{\downarrow}{}$ and $\underset{\uparrow}{}$ respectively) share GHZ states and are dis-entangled from the remaining qubits. Figure \ref{fig:entdist} illustrates this procedure.
\end{example}

\begin{figure}
    \centering
    \begin{subfigure}[b]{0.45\textwidth}
    \includegraphics[width=\textwidth]{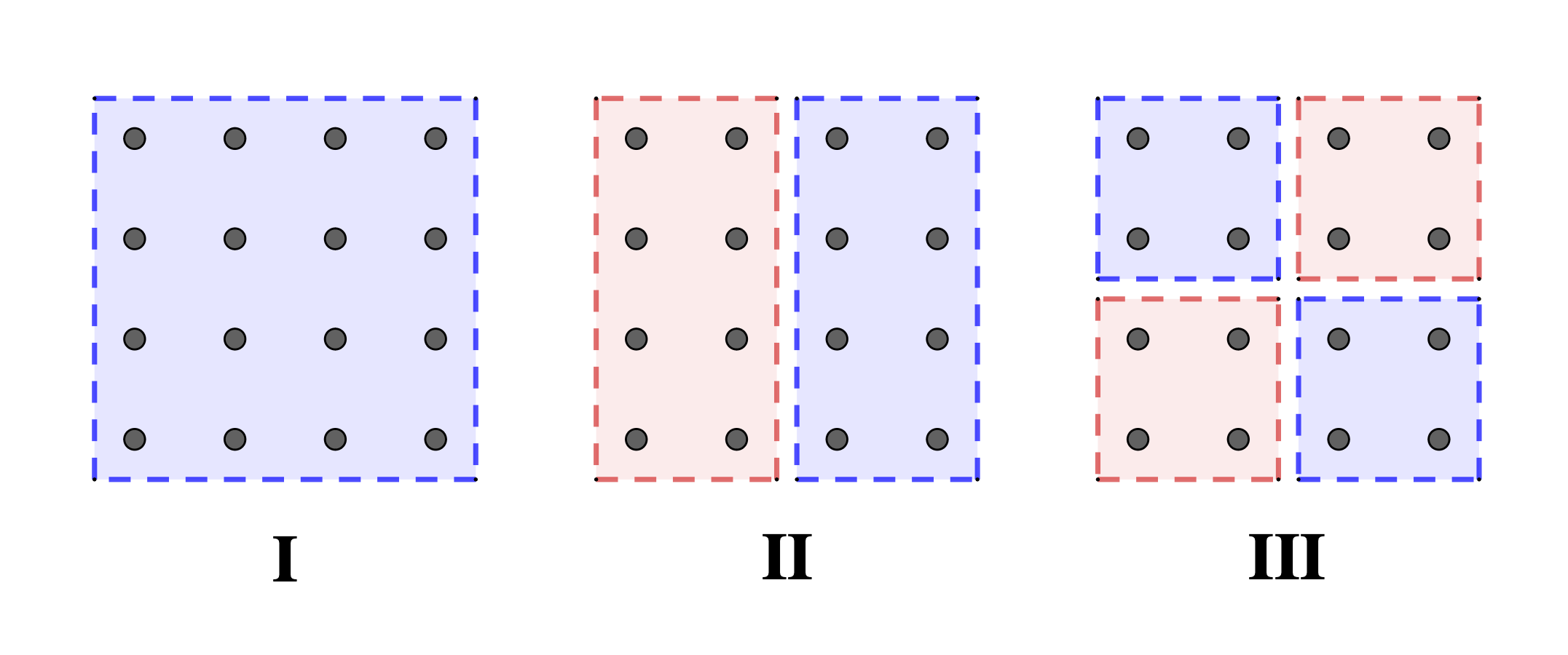}
    \caption{}
    \label{fig:staged}
  \end{subfigure}
  \hfill
  \begin{subfigure}[b]{0.45\textwidth}
    \includegraphics[width=\textwidth]{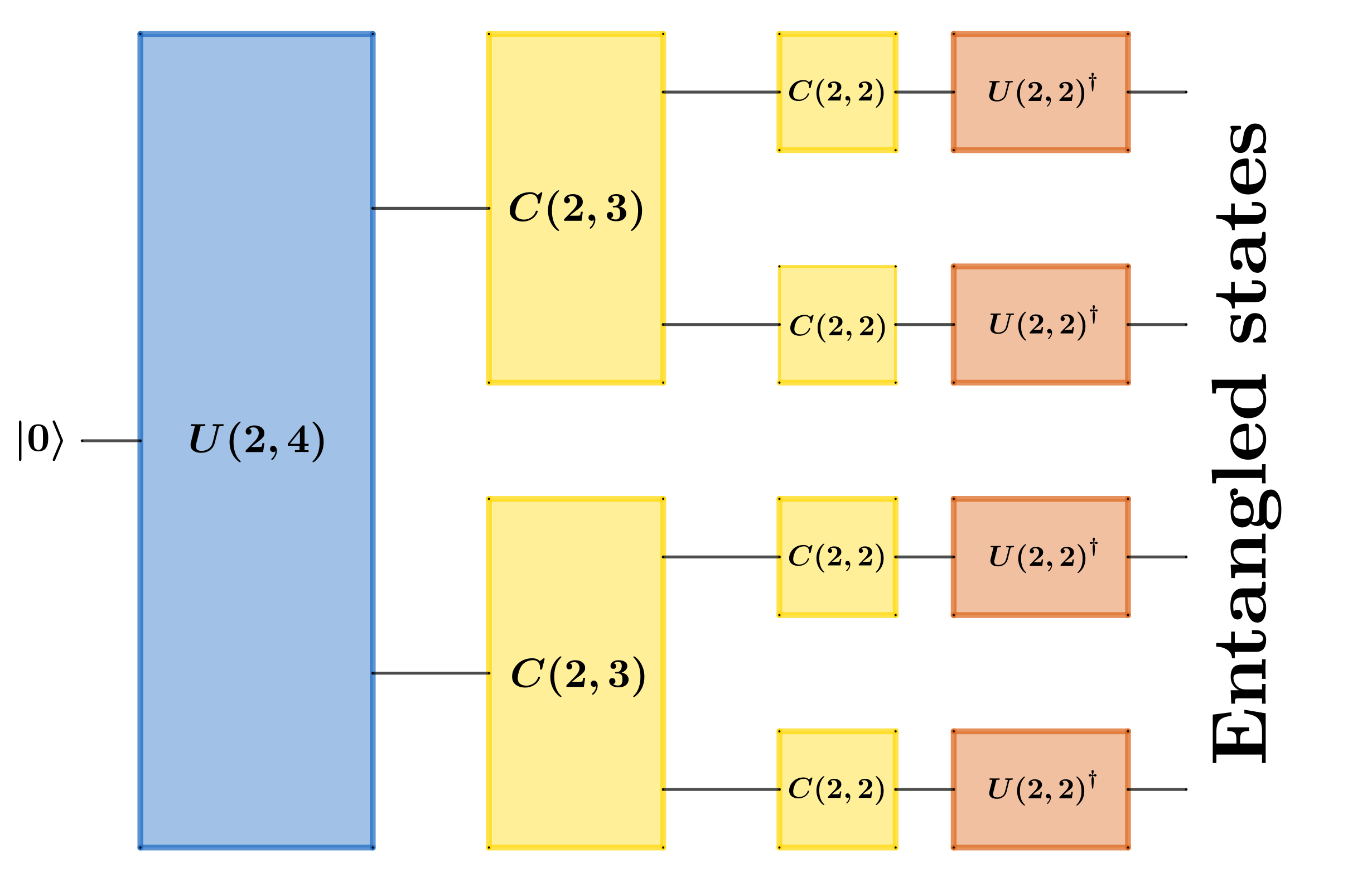}
    \caption{}
    \label{fig:entdist}
  \end{subfigure}
  \caption{(a) Hardware locality of staged decoding/encoding. After adding CNOT gates at each iteration across a partition of qubits, the subsequent gates are restricted to the corresponding subsets of qubits. (b) Staged entanglement distribution using $\ket{0}_{2, 4}$ to distribute $4-$party GHZ state. Initial state is obtained by encoding with $U(2, 4)$ (blue). At each stage, errors are corrected, by $C(2, 3), C(2, 2)$ (yellow). The entangled states are obtained by decoding on each set, using $U(2, 2)^\dagger$ (orange).}
  \label{fig:local_ent}
\end{figure}

\section{Conclusion}\label{sec:conc}
Using the Plotkin $(u, u + v)$ construction of classical RM codes, we showed that the QRM codeword can be written as a superposition of tensor products of QRM codewords of smaller length. Using these properties, we constructed recursive encoder-decoder circuits with lowered gate counts and logarithmic circuit depth. By estimating the entropy of entanglement in the codewords, we show that our encoders are efficient. We also provide a simple scheme to distill the entanglement. By reducing the necessary resources for encoding codes that admit transversal gates, we believe our work will assist in building fault tolerant quantum computers - through efficient practical coding schemes. In this article, through QRM codes we connect various ideas of encoding and entanglement, and we hope that this work brings new dimensions towards understanding how entanglement can be extracted from QECCs for multiparty communications/computations with such resources.

\section*{Code availability}
Python code for generating the circuits and estimate circuit properties can be found at \url{https://github.com/Praveen91299/QRM}. The functions that construct encoding circuits return the circuit, qubit permutation to be applied before the circuit, and the row index list as a dictionary object.

\section*{Acknowledgments}
P. Jayakumar acknowledges fellowship support from Kishore Vaigyanik Protsahan Yojana (KVPY) scheme, Government of India, for their undergraduate research. P. Jayakumar thanks Arpit Behera for inputs and insightful discussions while writing their undergraduate thesis at IISc, Bangalore.

\bibliographystyle{unsrt}
\bibliography{QRM}

\newpage
\appendix
\section{Properties of evaluation vectors and RM codewords}\label{app:RMprop}
The RM codewords can be constructed as evaluation vectors of multivariate polynomials. Many useful properties were proved in the Appendices of our previous work\cite{Nadkarni2024entanglement} and are directly referenced here. We present a few additional useful properties utilized here.

\begin{lemma}[Uniqueness of leading entry bit]\label{lem:uniquefirst}
    The index of the leading entry bit of the canonical generators of $\RM(r, m)$, i.e, the leading entry bit of each row in $G(r, m)$ is unique.
\end{lemma}
\begin{proof}
    We prove by induction.
    Due to the structure of the generator matrix in \eqref{def:Grm}, the leading entry is unique if it is unique for $G(r, m-1)$ and $G(r-1, m-1)$.
    It is sufficient to show that the lemma holds for $G(0, m)$ and $G(1, m)$ for any $m \geq 0$.
    $G(0, m)$ is a single row, thus the lemma is trivially true. The rows of $G(1, m)$ are of the form $\Eval{m}{x_i}$ for $i \in\{1, 2, \cdots, m\}$. The first evaluation point where $\Eval{m}{x_i}$ is $1$ (leading entry bit) is $\bm \rho(s)$ for $s \in \integers{1}{2^m}$ such that $\rho_i{(s)} =1, \rho_j{(s)} =0 \forall j \in \{1, 2, \cdots, m\}\setminus \{i\}$. By construction $\bm \rho{(s)}\neq \bm \rho{(t)}$ for $s\neq t$, thus the leading entries of $G(1, m)$ are unique.
\end{proof}

\begin{lemma}[Weight of evaluation vectors]\label{lem:evalwt}
    For $A, B \subseteq \integers{1}{m}, A\cap B =\phi$, the weight of the evaluation vectors of monomials have the following properties:
    \begin{itemize}
        \item[(i)] $\wt{\Eval{m}{x_A}} = 2^{m - |A|}$
        \item[(ii)] $\wt{\Eval{m}{x_Ax_i}} = \wt{\Eval{m}{x_A(1 + x_i)}} = 2^{m - |A| -1}$ for $i\in \integers{1}{m}\setminus A$.
        \item[(iii)] $\wt{\Eval{m}{x_A\prod_{i\in B}\left(1 + x_i\right)}} = 2^{m - |A| - |B|}$
    \end{itemize}
\end{lemma}
\begin{proof}
We use the bit indexing convention $\text{bin}(j-1) = \bm \rho(j) = (\rho_1(j), \rho_2(j), \cdots, \rho_m(j))$.
\begin{itemize}
    \item[(i)] $\Eval{m}{x_A}_j = 1$ iff $\rho_i(j) =1 \forall i \in A$. There are $2^{m-|A|}$ such evaluation points $\bm \rho(j), j \in \integers{1}{2^m}$.
    \item[(ii)] $\Eval{m}{x_Ax_i} = \Eval{m}{x_{A\cup \{i\}}}$. From (i), $\wt{\Eval{m}{x_{A\cup \{i\}}}} = 2^{m - |A\cup \{i\}|} = 2^{m - |A| - 1}$. $\Eval{m}{x_A(1 + x_i)}_j = 1$ iff $\rho_k(j) = 1 \forall k \in A$ and $\rho_i(j) = 0$. There are $2^{m - |A| - 1}$ such evaluation points.
    \item[(iii)] By using (ii) successively, we obtain $\wt{\Eval{m}{x_A\prod_{i\in B}\left(1 + x_i\right)}} = \wt{\Eval{m}{x_Ax_B}} = \wt{\Eval{m}{x_{A\cup B}}} = 2^{m - |A| - |B|}$.
\end{itemize}
\end{proof}

Since $\RM(r, m)$ consists of monomials upto degree $r$, the lowest Hamming weight of the generators in $\RM(r, m)$ is $2^{m-r}$. From the lemmas above, notice that by forming multinomials given by invertible linear combinations of monomials, we can obtain generators of lower Hamming weights. Since the degree is bounded by $r$, we cannot lower the Hamming weights of $G(r, m)$ beyond $2^{m-r}$.

\begin{lemma}(Form of evaluation vectors)\label{lem:evalform}
    Let $f\in \mathbb{F}_2[x_1,\cdots, x_m], g \in \mathbb{F}_2[x_2,\cdots,x_m]$ be multivariate polynomials of degree $\leq r$ and $\leq r-1$ respectively, such that $\Eval{m}{f} = \bm c$ and $\Eval{m-1}{g} = \bm u$. The evaluation vector $\bm c$ then has the form
    \begin{itemize}
        \item[(i)] $(\bm 0, \bm u)$ if $f = x_1g$
        \item[(ii)] $(\bm u, \bm u)$ if $f = g$
        \item[(iii)] $(\bm u, \bm 0)$ if $f = (1 + x_1)g$
    \end{itemize}
\end{lemma}
\begin{proof}
    (i), (ii) are trivial from constructions, (iii) = (i) + (ii).
\end{proof}
The above lemma extends to punctured codewords. An immediate consequence of Lemma \ref{lem:evalform} (iii) is if $\bm u = \Eval{m}{x_1x_2\cdots x_m}$, then $(\bm u, \bm 0) = \Eval{m+1}{(1 + x_1)x_2x_3\cdots x_m}$.

\section{Degeneracy of Quantum Reed Muller code}\label{app:degeneracy}
A quantum code is said to be degenerate if at least two correctable errors have the same syndrome. Non-degeneracy ensures that the decoded state is unique. In this Appendix, we prove that $\QRM(r, m)$ is a non-degenerate quantum code.

On performing syndrome measurement, the error in the state collapses to a Pauli error acting on the encoded state. Using the symplectic isomorphism, this error can be written as a $2N$ length binary vector, say $\bm e$. Then the syndrome $\bm s$ is related to $\bm e$ by $\bm s = \bm e J H^T$, where $J = \begin{pmatrix} 0 & 1 \\ 1 & 0 \end{pmatrix} \otimes \mathbb I_{N\times N}$ exchanges the $X$ and $Z$ errors, i.e., $e = (e_x, e_z), eJ = (e_z, e_x)$. Suppose errors $\bm e_1$ and $\bm e_2$ have the same syndrome, then 
\begin{align}
    &\bm e_1 J H^T = ~\bm e_2 J H^T\\
    \Rightarrow &(\bm e_1 - \bm e_2) J H^T = (\bm e_1 + \bm e_2) J H^T= ~\bm 0
\end{align}
This implies that $(\bm e_1 + \bm e_2)J = \bm e'$ belongs to the null space of $H$. 
For a CSS code with the parity check matrix $H$ of the form in Eq. \eqref{paritycheck}, the null space is spanned by rows of the matrix 
\begin{equation}
    \begin{pmatrix}
        G_1 & 0 \\
        0 & G_2
    \end{pmatrix}
\end{equation}
For a $\QRM(r, m)$ code, the minimum weight of the rowspace of the above matrix is $2^{m-r}$ as the minimum distances of the classical RM codes used in the QRM code are $2^{m-r}$ each. If the errors are correctable, then $\wt{\bm e_1}, \wt{\bm e_2} \leq \lfloor (d-1)/2\rfloor = 2^{m-r-1}-1$ as $d = 2^{m-r}$. Thus $\wt{\bm e'} = \wt{\bm e_1 \oplus \bm e_2} \leq \wt{\bm e_1} + \wt{\bm e_2} = 2^{m-r} - 2$ and $\bm e'$ cannot belong to the null space of $H$. Thus two correctable errors can never have the same syndrome and the code $\QRM(r, m)$ is non-degenerate.

\section{Recursive constructions of codewords and encoders }\label{app:decomposition}
We showed that $\QRM(r, m)$ codewords can be written as a superposition of tensor product of codewords of $\QRM(r, m-1)$, and provided recursively constructed encoders, $U(r, m)$. In this Appendix, we show similar constructions for punctured, zero rate, and punctured-zero rate quantum Reed-Muller codes.

The Corollaries \ref{cor:qubitsetdecomp} and \ref{cor:decomps} are applicable for Theorems \ref{thm:decompp}, \ref{thm:decompz}, and \ref{thm:decomppz} below and are not explicitly stated.

\subsection{Punctured Quantum Reed-Muller code}\label{app:punc}

\begin{theorem}[Punctured QRM codeword decomposition]\label{thm:decompp}
    For $r< m-1$ and $r \geq \lceil \frac{m-1}{2}\rceil$, the basis codewords of $\pQRM(r, m)$ on $2^{m}-1$ qubits provided in \eqref{def:pQRM} can be written as a superposition of tensor product of codewords of $\pQRM(r, m-1)$ and $\QRM(r, m-1)$ on the first $2^{m-1}-1$ qubits and last $2^{m-1}$ qubits respectively.

    For $\bm w = (\bm w_1, \bm w_2) \in \pRM(r, m)$ with $\bm w_1 \in \pRM(r, m-1),\bm w_2 \in \RM(r, m-1)$,
    \begin{equation}\label{decompp}
        \ket{\bm{w}}_{r, m}^* = N \sum_{ \bm{u} \in B} \ket{\bm{w_1} \oplus \bm u^*}_{r, m-1}^*
        \ket{\bm{w_2} \oplus \bm{u}}_{r, m-1} 
    \end{equation}
    where $B = \RM(r-1, m-1)^\perp /\RM(r, m-1)^\perp$ is a quotient set and $N$ is the normalization factor, $N = \frac{1}{\sqrt{|B|}}$.
\end{theorem}
\begin{proof}
    Recall the codewords of $\pQRM(r, m)$ in Eq, \eqref{def:pQRM} as
    $$
    \ket{\bm w}_{r, m}^* = \frac{1}{\sqrt{|\pRM(r, m)^\perp|}}\sum_{\bm c \in \pRM(r, m)^\perp}\ket{\bm w \oplus \bm c}
    $$
    for $\bm w \in \pRM(r, m)$ and $N_1 = 1/\sqrt{|\pRM(r, m)^\perp|}$. Rewriting the codewords $\bm c$ and $\bm w$ with the Plotkin $(u, u+v)$ construction for the punctured RM code, 
    \begin{align}
    N_1\sum_{\mathclap{\bm c \in \pRM(r, m)^\perp}}\ket{\bm w \oplus \bm c}
    &= N_1\sum_{\bm \alpha \in C}\sum_{\bm \beta \in D}\ket{\bm w_1\oplus\bm \alpha^*}\ket{\bm w_2\oplus\bm \alpha \oplus \bm \beta}\label{decomppstep0}\\
    &= N_1\sum_{\bm \alpha \in C}\left(\ket{\bm w_1\oplus\bm \alpha^*}\sum_{\bm \beta \in D}\ket{\bm w_2\oplus\bm \alpha \oplus \bm \beta}\right)\label{decomppstep1}\\
    &= N_2\sum_{\bm \alpha \in C}\ket{\bm w_1\oplus\bm \alpha^*}\ket{\bm w_2\oplus\bm \alpha}_{r, m-1}\label{decomppstep2}
    \end{align}
    where $C = \RM(r-1, m-1)^\perp/\{\bm 1\}$, $D = \RM(r, m-1)^\perp$,  $\bm w_1 \in \pRM(r, m), \bm w_2 \in \RM(r, m-1)$, and the normalization constant $N_2 = 1\Big/\sqrt{|\RM(r-1, m-1)^\perp/\{\bar{\bm 1}\}|}$ are the normalization constants.

    By Property \eqref{prop:RMsubgroup}. we have $\RM(r, m-1)^\perp \subset \RM(r-1, m-1)^\perp$. We then express $\bm \alpha \in \RM(r-1, m-1)^\perp/\{\bm 1\}$ as $\bm \alpha = \bm u \oplus \bm \beta$ for $\bm u \in \RM(r-1, m-1)^\perp /\RM(r, m-1)^\perp \equiv B$ and $\bm \beta \in \RM(r, m-1)^\perp /\{\bm 1\}$. Equation \eqref{decomppstep2} is rewritten as
    \begin{align}
        N_1\sum_{\bm c}\ket{\bm w \oplus \bm c} 
        &= N_2\sum_{\bm u\in B}\sum_{\mathrlap{\!\bm \beta \in \RM(r, m-1)^\perp/\{\bar{\bm 1}\}}}\ket{\bm w_1\oplus\bm u^* \oplus \bm \beta^*}\ket{\bm w_2\oplus\bm u \oplus \bm \beta}_{r, m-1}\label{decompp:step:1}\\
        &= N_2\sum_{\bm u\in B}\sum_{\mathrlap{\!\bm \beta \in \RM(r, m-1)^\perp/\{\bar{\bm 1}\}}}\ket{\bm w_1\oplus\bm u^* \oplus \bm \beta^*}\ket{\bm w_2\oplus\bm u}_{r, m-1}\label{decompp:step:2}\\
        &= N\sum_{\bm u\in B}\ket{\bm w_1\oplus\bm u^*}_{r, m-1}^*\ket{\bm w_2\oplus\bm u}_{r, m-1}\label{decompp:step:3}
    \end{align}
    where $N = 1\Big/ \sqrt{B} = N_2^2/N_1$. Equation \eqref{decompp:step:2} is obtained by Properties \eqref{CSS:addx} and \eqref{CSS:Xinv}, and Eq. \eqref{decompp:step:3} is obtained using the definition of the punctured QRM codewords. The parameter bound $r \geq \lceil (m-1)/2\rceil$ follows from the CSS code requirement and $r< m-1$ is since $\ket{\bm w}_{r, r}$ cannot be punctured to yield a valid $\mathrm{CSS}(C_1, C_2)$ codeword $\ket{\bm w}_{r, r}^*$ as the generator matrix of $C_2$, $G(r, r)^*$ is rank deficient.
\end{proof}

\subsubsection{Constructing recursive encoders}
We denote the encoding circuit for $\pQRM(r, m)$ by $U^*(r, m)$.  We construct the encoder $U^*(r, m)$ recursively using $U^*(r, m-1)$ and $U(r, m-1)$. This algorithm which we refer to as $\texttt{RecursivePQRM(r, m)}$ is based on the same idea as that of \texttt{RecursiveQRM(r, m)}. The construction differs in Steps $4$, $5$, and $6$, where $2^{m-1}$ is replaced with $2^{m-1}-1$, and $U^*(r, m-1)\otimes U(r, m-1)$ is added in Step $6$. The recursion can be continued till $U^*(r, r+1)$ which encodes states of the form $\ket{\bm w}_{r, r+1}^*=\ket{\bm w^*}$ which is a classical computational basis state for $\bm w \in \RM(r, r+1)$.

\textbf{Recursive encoder for $U^*(r, r+1)$}:
The punctured generator matrix, $G(r, r)^*$, is rank deficient. The classical code $\RM^*(r, r)$ cannot be utilized to faithfully encode information. Thus, a unitary encoder for $U^*(r, r)$ does not exists and $U^*(r, r+1)$ cannot be constructed recursively as earlier using $U^*(r, r)$. Instead, we recursively construct $U^*(r, r+1)$ with $U^*(r-1, r)$ and $U^c(r-1, r)$.

Observe that the codeword $\bm w \in \RM(r, r+1)^*$ can be written as
\begin{equation}
    \bm w^* = (\bm w_1^*, \bm w_1 \oplus \bm w_2) \oplus (\bm v^*, \bm v)
\end{equation}
where codewords $\bm w_1 = \bm m_{\bm w_1} G(r-1, r) \in \RM(r-1, r), \bm w_2 = \bm m_{\bm w_2} G(r-1, r) \in \RM(r-1, r)$, and $\bm v = m_{\bm v}[00\dots 01] \in \{00\dots 00, 00\dots 01\}$ are encoded by message bits/strings $\bm m_{\bm w_1}, \bm m_{\bm w_2} \in \mathbb F_2^{2^m - 1}$ and $m_{\bm v} \in \mathbb F_2$. Since $\bm w_1^* = \bm m_{\bm w_1}G(r-1, r)^*$, we have $\bm m_{\bm w_1} \equiv \bm m_{\bm w_1^*}$.
Additionally, notice that $U^c(r-1, r)\ket{\bm v} = U^c(r-1, r)^\dagger \ket{\bm v} = \ket{\bm v}$. This suggests encoding $\bm w_1^*$, followed by $(\bm v, \bm v)$, and then finally $\bm w_1\oplus \bm w_2$. 

Begin with the state $\ket{\bm m_{\bm w^*}}$ for $\bm m \in \mathbb F_2^{2^{r+1}-1}$. We apply a permutation $P(p)$ to obtain $\ket{m_{\bm w_1^*}, m_{\bm w_2}, m_v}$. Add CNOT gates to obtain $\ket{\bm m_{\bm w_1^*}, \bm m_{\bm w_1^*}\oplus \bm m_{\bm w_2}, m_v} \equiv \ket{\bm m_{\bm w_1^*}, \bm m_{\bm w_1 \oplus \bm w_2}, m_v}$. Applying $U^*(r-1, r)$ on the first $2^r-1$ qubits, we obtain $\ket{\bm w_1^*, \bm m_{\bm w_1 \oplus \bm w_2}, m_v}$. Add $\CNOT{2^{r+1}-1}{2^r-1}$ to obtain $\ket{(\bm w_1^*, \bm m_{\bm w_1 \oplus \bm w_2}, 0) \oplus (\bm v, \bm v)}$. Finally, add $U^c(r-1, r)$ on the last $2^r$ qubits to prepare $\ket{(\bm w_1^*, \bm w_1 \oplus \bm w_2) \oplus (\bm v^*, \bm v)} \equiv \ket{\bm w^*}$.

The Algorithm to construct $U^*(r, r+1)$ is provided in Figure \ref{alg:rec_punc_qrm}. Figure \ref{fig:ur_r1} illustrates this construction.

\begin{figure}
    \centering
    \includegraphics[width = 0.8\textwidth]{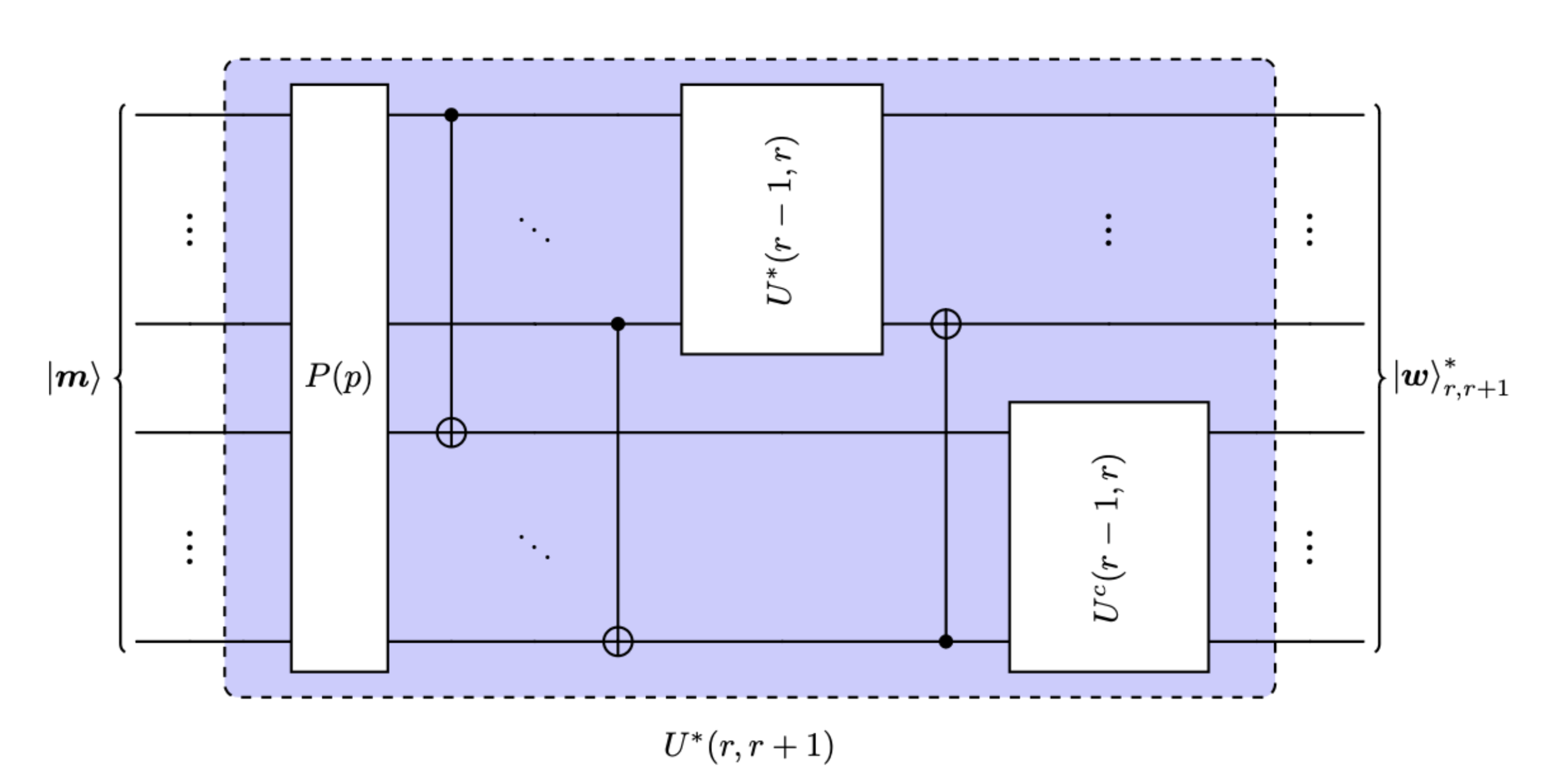}
    \caption{Recursive construction of $U^*(r, r+1)$.}
    \label{fig:ur_r1}
\end{figure}

\textit{CNOT Gate counts}: From the construction, it can be seen that the construction of $U^*(r, m)$ requires
\begin{equation}
    \zeta^*(r, m) = \begin{cases}
    \sum_{i=m-r-1}^r {m-1\choose i} + \zeta^*(r, m-1) + \zeta(r, m-1), &\text{if } r < m-1\\
    2^r + \zeta^*(r-1, r) + \zeta^c(r-1, r) & \text{if } m = r + 1 \\
    0 & \text{if } m = 1
   \end{cases}
\end{equation}
CNOT gates where $\zeta^*(r, m)$ denotes the CNOT gate count for the punctured code $\QRM(r, m)^*$. Since $G(0, 1)^* = \left[\begin{array}{c}1\end{array}\right]$, $C_{0, 1}^* = 0$ and $\zeta^*(0, 1) = 0$. 

\textit{Circuit depth}: On close inspection\footnote{The CNOT added in Step 3 of the construction of  $U^*(r, r+1)$ can be moved before $U(r-2, r-1)$ that appears in the recursive construction of $U^*(r-1, r)$ and likewise.}, the encoder $U^*(r, r+1)$ has $3r$ depth and the encoding circuit $U^*(r, m)$ has a $O(m + 2r)= O(\log n)$ depth.

\begin{figure}
    \centering
    \begin{mdframed}
        \centering{\textbf{\texttt{RecursivePQRM(r, m)}}\\
        Returns Encoder $U^*(r, m)$, qubit index map $M$ for $(m-1)/2\leq r < m$.}\\
    \begin{minipage}[t]{\textwidth}
        \centering{\underline{Case 1:} For $m > r + 1$:}
        \begin{description}[itemsep=0.5pt]
            \item[1.] Obtain $U^*(r, m-1)$, $M_1$ from \texttt{RecursivePQRM(r, m-1)} and $U(r, m-1)$, $M_2$ from \texttt{RecursiveQRM(r, m-1)}.
            \item[2.] Initialize circuit over $2^m-1$ qubits. Create $M = \rim(G)$ for
            $$G = \left[\begin{array}{c} G_1^* \\ \hline G_2^* \vert G_2\end{array}\right], ~~~\text{for } \begin{array}{c} G_1 = G(r, m)\setminus G(m-r-1, m) \cup \{\bm 1\}\\ G_2 = G(m-r-1, m-1)\setminus G(m-r-2, m-1)\end{array}$$
            \item[3.] Add gates $H(K), K = (k+1, \dots, k + {m-1\choose m-r-1})$ for $k = 1 + \sum_{i=m-r}^r{m\choose i}$.
            \item[4.] Add qubit permutation $P(p)$ such that $$p(M[(\bm u^*, \bm u)]) = M_1[\bm u^*] ~~\forall \bm u \in G_3,$$ $$p(M[(\bm 0, \bm v)]) = 2^{m-1}-1 + M_2[\bm v] ~~\forall \bm v\in G_4$$
            where $G_3 = G(r, m-1)\setminus G(m-r-2, m-1) \cup \{\overline{\bm 1}\}$\\ and $G_4 = G(r-1, m-1)\setminus G(m-r-2, m-1)$.
            \item[5.] Add $\CNOT{i}{j}$ for $i = M_1[\bm u^*]$, $j = 2^{m-1}-1 + M_2[\bm u]$ $\forall \bm u \in G(r, m-1)\setminus G(m-r-2, m-1)$.
            \item[6.] Add $U^*(r, m-1)\otimes U(r, m-1)$ to the circuit.
            \item[7.] Return circuit, $M$.
        \end{description}
    \end{minipage}
    \vspace{5pt}
    \vspace{5pt}
    \begin{minipage}[t]{\textwidth}
        \centering{\underline{Case 2:} For $m = r+1$:}
        \begin{description}
            \item[1.] Obtain $U^*(r-1, r)$, $M_1$ from $\texttt{RecursivePQRM(r-1, r)}$ and $U^c(r-1, r)$, $M_2$ from $\texttt{RecursiveBasisQRM(r-1, r)}$.
            \item[2.] Initialize circuit over $2^{r+1}-1$ qubits. Create $M = \rim(G(r, r+1)^*)$.
            \item[3.] Add qubit permutation $P(p)$ such that 
            $$p(M[(\bm u^*, \bm u)]) = M_1[\bm u^*] ~~\forall \bm u \in G(r-1, r),$$
            $$p(M[(\bm 0, \bm v)]) = 2^{r}-1 + M_2[\bm v] ~~\forall \bm v\in G(r-1, r)$$
            $$p(M[(\bm u^*, \bm u)]) = 2^{r+1}-1 ~~ \text{for } \bm u = 00\dots 01
            $$
            \item[4.] Add $\CNOT{i}{j}$ for $i = M_1[\bm u^*]$, $j = 2^{r}-1 + M_2[\bm u] \forall \bm u \in G(r-1, r)$.
            \item[5.] Add $U^*(r-1, r)$ on the first $2^r -1$ qubits.
            \item[6.] Add $\CNOT{i}{j}$ for $i = 2^{r+1}-1$, $j = 2^r - 1$.
            \item[7.] Add $U^c(r-1, r)$ on the last $2^{r}$ qubits.
            \item[8.] Return circuit, $M$.
        \end{description}
    \end{minipage}
    \end{mdframed}
    \caption{Recursive construction for $U^*(r, m)$.}
    \label{alg:rec_punc_qrm}
\end{figure}

\textbf{Encoder for punctured computational basis state QRM}: The encoding circuit $U^{(*, c)}(r, m)$ that prepares the state $\ket{\bm w^*}, \bm w^* \in \RM(r, m)^*$ can be constructed similar to the construction of $U^c(r, m)$, for recursions up till $m = r + 1$. $U^*(r, r+1)$ is then used to encode $\ket{\bm w^*}, \bm w \in \RM(r, r+1)^*$. This algorithm, $\texttt{RecursiveBasisPQRM(r, m)}$ is provided in Figure \ref{alg:rec_basis_punc_qrm}.
From the construction, the CNOT gate counts are given by the recursive relations:
\begin{align}
    \zeta^{(*, c)}(r, m) = \begin{cases} \sum_{i=0}^{r} {m-1\choose i} + \zeta^{(*, c)}(r, m-1) + \zeta^c(r, m-1) & \text{if } r < m-1\\
    \zeta^*(r, r+1) & \text{if } m = r+1
    \end{cases}
\end{align}
\begin{figure}
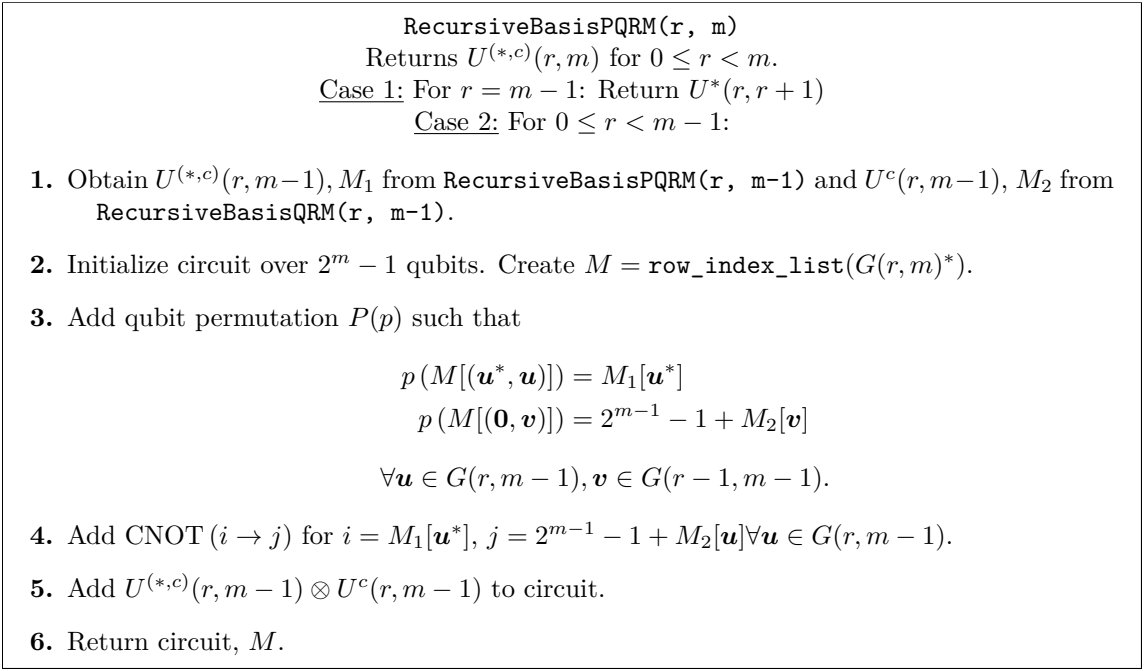

    \centering
\begin{mdframed}
    \centering{\texttt{RecursiveBasisPQRM(r, m)}\\
    Returns $U^{(*, c)}(r, m)$ for $0 \leq r < m$.}\\
    \centering{\underline{Case 1:} For $r = m-1$:} Return $U^{*}(r, r+1)$ \\
    \centering{\underline{Case 2:} For $0 \leq r < m - 1$:}\\
     \begin{description}
                \item[1.] Obtain $U^{(*, c)}(r, m-1), M_1$ from \texttt{RecursiveBasisPQRM(r, m-1)} and $U^c(r, m-1)$, $M_2$ from \texttt{RecursiveBasisQRM(r, m-1)}.
                \item[2.] Initialize circuit over $2^m-1$ qubits. Create $M = \rim(G(r, m)^*)$.
                \item[3.] Add qubit permutation $P(p)$ such that 
                \begin{align}
                    p\left(M[(\bm u^*, \bm u)]\right) &= M_1[\bm u^*]\nonumber\\
                    p\left(M[(\bm 0, \bm v)]\right) &= 2^{m-1}-1 + M_2[\bm v]\nonumber
                \end{align}
                $$ \forall \bm u \in G(r, m-1), \bm v\in G(r-1, m-1).$$        
                \item[4.] Add $\CNOT{i}{j}$ for $i = M_1[\bm u^*]$, $j = 2^{m-1}-1 + M_2[\bm u] \forall \bm u \in G(r, m-1)$.
                \item[5.] Add $U^{(*, c)}(r, m-1)\otimes U^c(r, m-1)$ to circuit.
                \item[6.] Return circuit, $M$.
    \end{description}
\end{mdframed}
    \caption{Recursive construction for $U^{(*, c)(r, m)}$}
    \label{alg:rec_basis_punc_qrm}
\end{figure}

\textbf{State preparation circuit for punctured QRM}:
We construct $U^{(*, s)}(r, m)$ that prepares the state $\ket{\bm w}_{r, m}^*, \bm w \in \RM(r, m)/\{\bar{\bm 1}\}$, where we explicitly do not include $\bar{\bm 1}$. Recall that such a circuit is useful when we wish to prepare encoded ancillary states $\ket{\bm 0}_{r, m}^*$.

We obtain an optimized construction by simply removing the generator $\bm 1$ from the generator matrix in the case of $r = m$. The construction is similar to the construction of $U^*(r, m)$, for recursions up till $m = r+1$. The algorithm to construct $U^{(*, s)}(r, m)$ for $m=r+1$ is provided in Figure \ref{alg:rec_punc_qrm_s}.

\begin{figure}
    \centering
    \begin{mdframed}
    \begin{subfigure}{\textwidth}
        \centering{\noindent\textbf{\texttt{RecursiveStatePrepPQRM(r, m)}}, for $(m-1)/2\leq r\leq m$\\
        Returns: Encoder $U^{(*, s)}(r, m)$, qubit index map $M$.}\\
        \underline{Case 1:} If $(m-1)/2\leq r<m-1$:
        \begin{description}[itemsep=0.5pt]
            \item[1.] Obtain $U^{(*, s)}(r, m-1)$, $M_1$ from \texttt{RecursiveStatePrepPQRM(r, m-1)} and $U(r, m-1)$, $M_2$ from \texttt{RecursiveQRM(r, m-1)}.
            \item[2.] Initialize circuit over $2^m-1$ qubits. Create $M = \rim(G)$ for
            $$G = \left[\begin{array}{c} G_1^* \\ \hline G_2^* \vert G_2\end{array}\right], ~~~\text{for } \begin{array}{c} G_1 = G(r, m)\setminus G(m-r-1, m)\\ G_2 = G(m-r-1, m-1)\setminus G(m-r-2, m-1)\end{array}$$
            \item[3.] Add gates $H(K), K = (k+1, \dots, k + {m-1\choose m-r-1})$ for $k = \sum_{i=m-r}^r{m\choose i}$.
            \item[4.] Add qubit permutation $P(p)$ such that $$p(M[(\bm u^*, \bm u)]) = M_1[\bm u] ~~\forall \bm u \in G_3,$$ $$p(M[(\bm 0, \bm v)]) = 2^{m-1}-1 + M_2[\bm v] ~~\forall \bm v\in G_4$$
            where $G_3 = G(r, m-1)\setminus G(m-r-2, m-1)$\\ and $G_4 = G(r-1, m-1)\setminus G(m-r-2, m-1)$.
            \item[5.] Add $\CNOT{i}{j}$ for $i = M_1[\bm u^*]$, $j = 2^{m-1}-1 + M_2[\bm u]$ $\forall \bm u \in G_3$.
            \item[6.] Add $U^{(s,*)}(r, m-1)\otimes U(r, m-1)$ to the circuit.
            \item[7.] Return circuit, $M$.
        \end{description}
    \end{subfigure}
    \vspace{5pt}
    \vspace{5pt}
    \begin{subfigure}{\textwidth}
        \begin{minipage}[t]{0.49\textwidth}
        \centering\underline{\hspace{ 2 in}}\\
            \centering{\underline{Case 2:} If $r = m-1$:}
            \begin{description}
                \item[1.] Obtain $U^{(*, s)}(r, r)$, $M_1$ from $\texttt{RecursiveStatePrepPQRM(r, r)}$ and $U(r, r)$, $M_2$ from $\texttt{RecursiveQRM(r, r)}$.
                \item[2.] Initialize circuit over $2^{r+1}-1$ qubits. Create $M =\rim(G(r, r+1)^*)$.
                \item[3.] Add qubit permutation $P(p)$ such that 
                $$p(M[(\bm u^*, \bm u)]) = M_1[\bm u^*] ~~\forall \bm u \in G(r, r)\setminus \{\bar 1\},$$
                $$p(M[(\bm 0, \bm v)]) = 2^{r}-1 + M_2[\bm v] ~~\forall \bm v\in G(r, r)$$        
                \item[4.] Add $\CNOT{i}{j}$ for $i = M_1[\bm u^*]$, $j = 2^{r}-1 + M_2[\bm u] \forall \bm u \in G(r, r)\setminus \{\bar 1\}$.
                \item[5.] Add $U^{(*, s)}(r, r) \otimes U(r, r)$ to the circuit.
                \item[6.] Return $U^{(*, s)}(r, r+1)$, $M$.
            \end{description}
        \end{minipage}
        \hspace{5pt}
        \vline
        \hspace{5pt}
        \begin{minipage}[t]{0.49\textwidth}
        \centering\underline{\hspace{ 2 in}}\\
            \centering{\underline{Case 3:} If $r = m$:}
            \begin{description}
                \item[1.] Obtain $U^{(*, s)}(r-1, r-1)$, $M_1$ from $\texttt{RecursiveStatePrepPQRM(r-1, r-1)}$ and $U(r-1, r-1)$, $M_2$ from $\texttt{RecursiveQRM(r-1, r-1)}$.
                \item[2.] Initialize circuit over $2^{r}-1$ qubits. Create $M = \rim\left(G(r, r)\setminus\{\bar 1\}\right)$.
                \item[3.] Add qubit permutation $P(p)$ such that 
                $$p(M[(\bm u, \bm u)]) = M_1[\bm u]$$
                $$p(M[(\bm 0, \bm v)]) = 2^{r-1}-1 + M_2[\bm v]$$
                $\forall \bm u \in G(r-1, r-1)\setminus \{\bar 1\}, \bm v\in G(r-1, r-1)$
                \item[4.] Add $\CNOT{i}{j}$ for $i = M_1[\bm u]$, $j = 2^{r-1}-1 + M_2[\bm u] \forall \bm u \in G(r-1, r-1)\setminus \{\bar 1\}$.
                \item[5.] Add $U^{(*, s)}(r-1, r-1) \otimes U(r-1, r-1)$.
                \item[6.] Return $U^{(*, s)}(r, r)$, $M$.
            \end{description}
        \end{minipage}
        
    \end{subfigure}
    \end{mdframed}
    \caption{Recursive construction for $U^{(*, s)}(r, m)$}
    \label{alg:rec_punc_qrm_s}
\end{figure}
\textit{CNOT gate count}: From the constructions, the CNOT counts for the state preparation circuits is
\begin{align}
    \zeta^{(*, s)}(r, m) =\begin{cases} \sum_{i = m-r-1}^r {m-1 \choose i}+ \zeta^{(*, s)}(r, m-1) + \zeta(r, m-1) & \text{if } r < m - 1\\
    2^r - 1 + \zeta^{(*, s)}(r, r) + \zeta(r, r) & \text{if } m = r+1\\
    2^{r-1} - 1 + \zeta^{(*, s)}(r-1, r-1) + \zeta(r-1, r-1) & \text{if } m = r\\
    0 & \text{if } m=r=1
    \end{cases}
\end{align}

\subsection{Zero rate Quantum Reed-Muller codes}\label{app:zero}

\begin{theorem}[Zero rate QRM codeword decomposition]\label{thm:decompz}
    For $-1\leq r\leq r_{\diamond}<m$, the basis codewords of $\zQRM(r, m; r_{\diamond}, m_{\diamond})$ on $2^{m}$ qubits provided in \eqref{def:zQRMgen} can be written as a superposition of tensor product of codewords of $\zQRM(r-1, m-1; r_{\diamond}, m_{\diamond})$ on the first $2^{m-1}$ qubits and last $2^{m-1}$ qubits.

    For $\bm w = (\bm w_1, \bm w_2) \in \RM(r_{\diamond}, m)$ with $\bm w_1,\bm w_2 \in \RM(r_{\diamond}, m-1)$,
    \begin{align}
        \ket{\bm{w}}_{r, m}^{(0)} = N &\sum_{ \bm{u} \in B} \ket{\bm{w_1} \oplus \bm{u}}_{r-1, m-1}^{(0)}
        \ket{\bm{w_2} \oplus \bm{u}}_{r-1, m-1}^{(0)} \label{decompz}
    \end{align}
    where $B = \RM(r, m-1)/\RM(r-1, m-1)$ is a quotient set and $N$ is the normalization factor, $N = \frac{1}{\sqrt{|B|}}$.
\end{theorem}
\begin{proof}
    Recall the codewords of $\zQRM(r, m; r_{\diamond}, m_{\diamond})$ in Eq. \eqref{def:zQRMgen} as
    \begin{equation}
        \ket{\bm w}_{r, m}^{(0)} = N_1\sum_{\mathclap{\bm c \in \RM(r, m)}}\ket{\bm w \oplus \bm c}
    \end{equation}
    for $\bm w\in \RM(r_{\diamond}, m)$ and $N_1 = 1\Big/\sqrt{|\RM(r, m)|}$. Rewriting the codewords $\bm c$ and $\bm w$ with the Plotkin $(u, u+v)$ construction,
    \begin{align}
        N_1\sum_{\mathclap{\bm c \in \RM(r, m)}}\ket{\bm w \oplus \bm c} =& N_1\sum_{\alpha\in C}\sum_{\beta \in D}\ket{\bm w_1 \oplus \bm \alpha}\ket{\bm w_2 \oplus \bm \alpha \oplus \bm \beta},\\
        =& N_1\sum_{\bm \alpha\in C}\left(\ket{\bm w_1 \oplus \bm \alpha}\sum_{\bm \beta \in D}\ket{\bm w_2\oplus \bm \alpha \oplus \bm \beta}\right),\\
        =& N_2\sum_{\bm \alpha \in C}\ket{\bm w_1 \oplus \bm \alpha}\ket{\bm w_2 \oplus \bm \alpha}_{r-1, m-1}^{(0)},\label{zqrmproofstep}
    \end{align}
    where $C = \RM(r, m-1)$, $D = \RM(r-1, m-1)$, $\bm w_1, \bm w_2 \in \RM(r_{\diamond}, m-1)$, and normalization constant $N_2 = 1\Big/\sqrt{|\RM(r, m-1)|}$.

    By Property \eqref{prop:RMsubgroup}, we have $\RM(r-1, m-1) \subset \RM(r, m-1)$. We then express $\bm \alpha\in \RM(r, m-1)$ in Eq. \eqref{zqrmproofstep} as $\bm \alpha = \bm u \oplus \bm \beta$ for $\bm u \in \RM(r, m-1)/\RM(r-1, m-1) := B, \bm \beta \in \RM(r-1, m-1) \equiv D$. Equation \eqref{zqrmproofstep} is rewritten as
    \begin{align}
        N_1\sum_{\mathclap{\bm c \in \RM(r, m)}}\ket{\bm w \oplus \bm c} 
        &= N_2\sum_{\bm u\in B}\sum_{\bm \beta \in D}\ket{\bm w_1\oplus\bm u \oplus \bm \beta}\ket{\bm w_2\oplus\bm u \oplus \bm \beta}_{r-1, m-1}\label{zqrm:decomp:step:1}\\
        &= N_2\sum_{\bm u\in B}\sum_{\bm \beta \in D}\ket{\bm w_1\oplus\bm u \oplus \bm \beta}\ket{\bm w_2\oplus\bm u}_{r-1, m-1}\label{zqrm:decomp:step:2}\\
        &= N\sum_{\bm u\in B}\ket{\bm w_1\oplus\bm u}_{r-1, m-1}^{(0)}\ket{\bm w_2\oplus\bm u}_{r-1, m-1}^{(0)}\label{zqrm:decomp:step:3}
    \end{align}
    where $N = 1\Big/ \sqrt{B} = N_2^2/N_1$. Equation \eqref{zqrm:decomp:step:2} is obtained by Properties \eqref{CSS:addx} and \eqref{CSS:Xinv} of CSS codewords, and Eq. \eqref{zqrm:decomp:step:3} is obtained using the definition of the QRM codewords. The bounds on the range of $r, r_{\diamond}$ follow from the CSS code requirements and $\RM(-1, m) = \phi$. 
\end{proof}

\subsubsection{Constructing recursive encoders}
Recall that the code states of $\zQRM(r, m; r, m)$ are of the form
\begin{align}
    \ket{\bm 0}_{r, m}^{(0)} =& \frac{1}{\sqrt{|\RM(r, m)|}}\sum_{\bm c \in \RM(r, m)}\ket{\bm c}
\end{align}
and can be encoded by a circuit consisting of Hadamard gates on qubits $\{q_1, \dots, q_k\}$ with $k = \sum_{i=0}^{r}{m \choose i}$, followed by the encoder $U^c(r, m)$. However, we present an alternative recursive construction based on Theorem \ref{thm:decompz}. This construction demonstrates lower CNOT gate count for certain $(r, m)$ values.

Consider the state $\ket{\bm w}_{r, m}^{(0)}$ in $\zQRM(r, m; r_{\diamond}, m_{\diamond})$ for $\bm w \in \RM(r_{\diamond}, m)$. From Theorem \ref{thm:decompz}, we obtain a decomposition
\begin{equation}
    \ket{\bm w}_{r, m}^{(0)} = \frac{1}{\sqrt{|D|}}\sum_{\bm u\in D}\ket{\bm w_1 \oplus \bm u}_{r-1, m-1}^{(0)}\ket{\bm w_2 \oplus \bm u}_{r-1, m-1}^{(0)}
\end{equation}
where $D = \RM(r, m-1)\setminus\RM(r-1, m-1)$ and $\bm w = (\bm w_1, \bm w_2)$ for $\bm w_1, \bm w_2 \in \RM(r_\diamond, m-1)$. 
As $\ket{\bm w_1 \oplus \bm u}_{r-1, m-1}^{(0)} \in \zQRM(r-1, m-1; r_{\diamond}, m_{\diamond})$, by following a construction similar to $\texttt{RecursiveQRM(r, m)}$ we can construct an encoder for $\zQRM(r, m; r_{\diamond}, m_{\diamond})$ by utilizing encoders for $\zQRM(r-1, m-1; r_{\diamond}, m_{\diamond})$. Algorithm \texttt{RecursiveZQRM(r, m; r$_{\diamond}$, m$_{\diamond}$)} in Figure \ref{alg:rec_zero} provides the recursive construction for the encoder which we denote by $U(r, m; r_{\diamond}, m_{\diamond})$, where $r_{\diamond} > r, m_{\diamond}> m$ are the initial $(r, m)$ parameters defining the state. We retain $(r_{\diamond}, m_{\diamond})$, the initial code parameters since we require the information about the dimensions of the message string $\bm m_{\bm w}$ at each iteration.

The codeword, on successive iterations of decompositions would have codewords of the form $\ket{\bm w}_{r_{\diamond}-i, m_{\diamond}-i}^{(0)}, \bm w \in \RM(r_{\diamond}, m_{\diamond} - i)$.
For the state to remain a valid codeword, we have the conditions
\begin{align}
    r_{\diamond} -i \geq -1\\
    m_{\diamond} - i \geq r_{\diamond}
\end{align}
Thus, the number of iterations is restricted to $i_{\max} = \min(r_{\diamond}+1, m_{\diamond} - r_{\diamond})$ since $\zQRM(r, m; r_{\diamond}, m_{\diamond})$ cannot be decomposed by Theorem \ref{thm:decompz} beyond $r =-1$ or $m = r_{\diamond}$. We have two distinct cases of $(r_{\diamond}, m_{\diamond})$ that satisfy these conditions:

\textit{(i) When $2r_{\diamond} + 1\leq m_{\diamond}$}: The last iteration of decomposition, $i_{\max}=r_{\diamond} + 1$, will end with classical bit-string states $\ket{\bm mG(r_{\diamond}, m_{\diamond} - r_{\diamond}-1)}$. 
Encoding this state is done recursively by the computational basis state encoder $U^c(r_{\diamond}, m_{\diamond} - r_{\diamond}-1)$ provided in Figure \ref{alg:rec_qrm_c}.

\textit{(ii) When $2r_{\diamond} + 1 > m_{\diamond}$}: We cannot decompose the codeword state beyond $(m_{\diamond} - r_{\diamond})$ iterations. For $i_{\max} =  (m_{\diamond} - r_{\diamond})$, the codeword is of the form $$\ket{\bm w}_{2r_{\diamond} - m_{\diamond}, r_{\diamond}}^{(0)} = N\sum_{\mathclap{\bm c \in \RM(2r_{\diamond} - m_{\diamond}, r_{\diamond})}} \ket{\bm w \oplus \bm c},$$ for $\bm w\in \RM(r_{\diamond}, r_{\diamond})$ and normalization constant $N = 1/\sqrt{|\RM(2r_{\diamond} - m_{\diamond}, r_{\diamond})|}$. It is straightforward to see that $k = \sum_{i=0}^{2r_{\diamond} - m_{\diamond}}{r_{\diamond} \choose i}$ Hadamard gates and the encoder $U(r_{\diamond}, r_{\diamond})$ encodes this state.

\textit{CNOT gate count}: Following the construction above, the CNOT gate count for the recursive construction of $U^{(0)}(r, m; r_{\diamond}, m_{\diamond})$ is given by
\begin{align}
    \zeta^{(0)}(r, m; r_{\diamond}, m_{\diamond}) =& \begin{cases}
        \zeta^c(r_{\diamond}, m) & 2r_{\diamond} + 1 \leq m_{\diamond} \text{ and } r = -1\\
        \zeta(r_{\diamond}, r_{\diamond}) & 2r_{\diamond} + 1 > m_{\diamond} \text{ and } m = r_{\diamond}\\
        \sum_{i = r}^{r_{\diamond}}{m-1\choose i} + 2\zeta^{(0)}(r-1, m-1; r_{\diamond}, m_{\diamond}) & \text{otherwise}
    \end{cases}
\end{align}

\begin{figure}
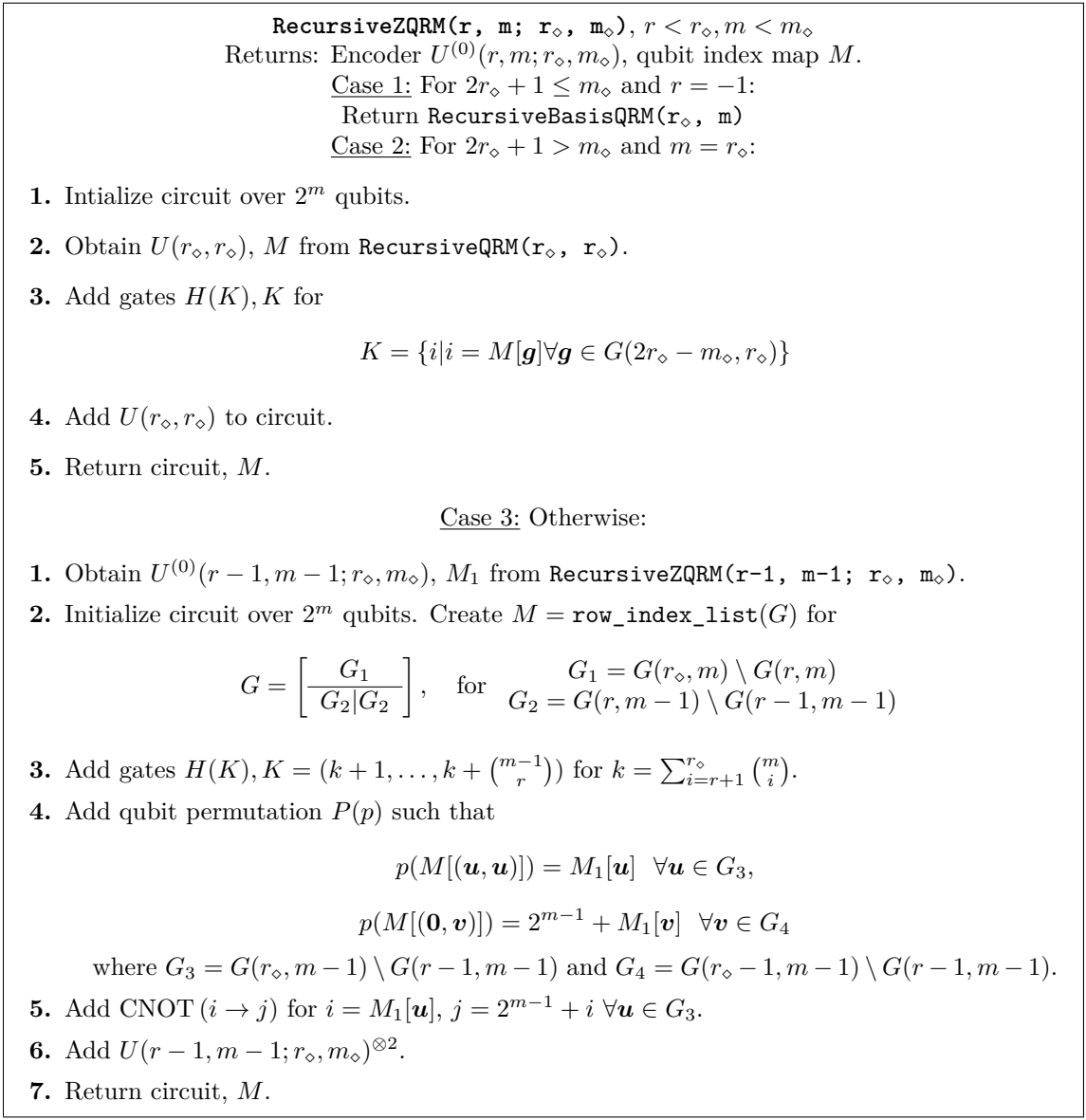
 
    \centering
    \begin{mdframed}
    \centering{\noindent\textbf{\texttt{RecursiveZQRM(r, m; r$_{\diamond}$, m$_{\diamond}$)}}, $r < r_{\diamond}, m < m_{\diamond}$\\
    Returns: Encoder $U^{(0)}(r, m; r_{\diamond}, m_{\diamond})$, qubit index map $M$.}
    \\\underline{Case 1:} For $2r_{\diamond} + 1 \leq m_{\diamond}$ and $r = -1$: \\Return \texttt{RecursiveBasisQRM(r$_{\diamond}$, m)}
    \\\underline{Case 2:} For $2r_{\diamond} + 1 > m_{\diamond}$ and $m = r_{\diamond}$:
    \begin{description}
        \item[1.] Intialize circuit over $2^m$ qubits.
        \item[2.] Obtain $U(r_{\diamond}, r_{\diamond})$, $M$ from \texttt{RecursiveQRM(r$_{\diamond}$, r$_{\diamond}$)}.
        \item[3.] Add gates $H(K), K$ for
        $$K = \{i \vert i = M[\bm g] \forall \bm g \in G(2r_{\diamond} - m_{\diamond}, r_{\diamond})\}$$
        \item[4.] Add $U(r_{\diamond}, r_{\diamond})$ to circuit.
        \item[5.] Return circuit, $M$.
    \end{description}
    \underline{Case 3:} Otherwise:\\
    \begin{description}[itemsep=0.5pt]
        \item[1.] Obtain $U^{(0)}(r-1, m-1; r_{\diamond}, m_{\diamond})$, $M_1$ from \texttt{RecursiveZQRM(r-1, m-1; r$_{\diamond}$, m$_{\diamond}$)}.
        \item[2.] Initialize circuit over $2^m$ qubits. Create $M = \rim(G)$ for 
         $$G = \left[\begin{array}{c} G_1 \\ \hline G_2 \vert G_2\end{array}\right], ~~~\text{for } \begin{array}{c} G_1 = G(r_{\diamond}, m)\setminus G(r, m) \\ G_2 = G(r, m-1)\setminus G(r-1, m-1)\end{array}$$
        \item[3.] Add gates $H(K), K = (k+1, \dots, k + {m-1\choose r})$ for $k = \sum_{i=r+1}^{r_{\diamond}}{m\choose i}$.
        \item[4.] Add qubit permutation $P(p)$ such that 
        $$p(M[(\bm u, \bm u)]) = M_1[\bm u] ~~\forall \bm u \in G_3,$$ 
        $$p(M[(\bm 0, \bm v)]) = 2^{m-1} + M_1[\bm v] ~~\forall \bm v\in G_4$$
        where $G_3 = G(r_{\diamond}, m-1)\setminus G(r-1, m-1)$ and $G_4 = G(r_{\diamond}-1, m-1)\setminus G(r-1, m-1)$.
        \item[5.] Add $\CNOT{i}{j}$ for $i = M_1[\bm u]$, $j = 2^{m-1} + i$ $\forall \bm u \in G_3$.
        \item[6.] Add $U(r-1, m-1; r_{\diamond}, m_{\diamond})^{\otimes 2}$.
        \item[7.] Return circuit, $M$.
    \end{description}
    \end{mdframed}
    \caption{Recursive construction for $U(r, m; r_{\diamond}, m_{\diamond})$.}
    \label{alg:rec_zero}
\end{figure}

\subsection{Punctured zero rate Quantum Reed-Muller codes}\label{app:punc_zero}
We repeat the previous analysis and constructions for the punctured zero rate code now.

\begin{theorem}[Punctured zero rate QRM codeword decomposition]\label{thm:decomppz}
    For $0 < r\leq r_{\diamond}<m$, the basis codewords of $\pzQRM(r, m; r_{\diamond}, m_{\diamond})$ on $2^{m}-1$ qubits can be written as a superposition of tensor product of codewords of $\pzQRM(r-1, m-1; r_{\diamond}, m_{\diamond})$ on the first $2^{m-1}-1$ qubits and $\pzQRM(r-1, m-1; r_{\diamond}, m_{\diamond})$ on the last $2^{m-1}$ qubits.

    For $\bm w = (\bm w_1, \bm w_2) \in \pRM(r_{\diamond}, m)$ with $\bm w_1 \in \pRM(r_{\diamond}, m-1)$ and $\bm w_2 \in \RM(r_{\diamond}, m-1)$,
    \begin{align}
        \ket{\bm{w}}_{r, m}^{(*, 0)} = N &\sum_{ \bm{u} \in B} \ket{\bm{w_1} \oplus \bm{u}^*}_{r-1, m-1}^{(*, 0)}
        \ket{\bm{w_2} \oplus \bm{u}}_{r-1, m-1}^{(0)} \label{decomppz}
    \end{align}
    where $B = \RM(r, m-1)/\RM(r-1, m-1)$ is a quotient set and $N$ is the normalization factor, $N = \frac{1}{\sqrt{|B|}}$.
\end{theorem}
\begin{proof}
    Recall the codewords of $\pzQRM(r, m; r_{\diamond}, m_{\diamond})$ in Eq. \eqref{def:pzQRM} as
    $$\ket{\bm w}_{r, m}^{(*, 0)} = N_1\sum_{\mathclap{\bm c \in \pRM(r, m)/\{\bm 1\}}}\ket{\bm w \oplus \bm c}$$ for $\bm w \in \pRM(r_{\diamond}, m)$ and the normalization constant $N_1 = 1/\sqrt{|\pRM(r, m)/\{\bm 1\}|}$. Rewriting the codewords $\bm c$ and $\bm w$ with the Plotkin $(u, u+v)$ construction,

\begin{align}
        N_1\sum_{\mathclap{\bm c \in \pRM(r, m)/\{\bm 1\}}}\ket{\bm w \oplus \bm c} =& N_1\sum_{\alpha\in C}\sum_{\beta \in D}\ket{\bm w_1 \oplus \bm \alpha^*}\ket{\bm w_2 \oplus \bm \alpha \oplus \bm \beta},\\
        =& N_1\sum_{\bm \alpha\in C}\left(\ket{\bm w_1 \oplus \bm \alpha^*}\sum_{\bm \beta \in D}\ket{\bm w_2\oplus \bm \alpha \oplus \bm \beta}\right),\\
        =& N_2\sum_{\bm \alpha \in C}\ket{\bm w_1 \oplus \bm \alpha^*}\ket{\bm w_2 \oplus \bm \alpha}_{r-1, m-1}^{(0)},\label{pzqrmproofstep}
    \end{align}
    where $C = \RM(r, m-1) /\{\bm 1\}$, $D = \RM(r-1, m-1)$, $\bm w_1 \in \pRM(r_{\diamond}, m-1), \bm w_2 \in \RM(r_{\diamond}, m-1)$, and normalization constant $N_2 = 1\Big/\sqrt{|\RM(r, m-1)/\{\bm 1\}|}$.

    By Property \eqref{prop:RMsubgroup}, we have $\RM(r-1, m-1) \subset \RM(r, m-1)$ and we express $\bm \alpha\in \RM(r, m-1)/\{\bm 1\}$ in Eq. \eqref{pzqrmproofstep} as $\bm \alpha = (\bm u \oplus \bm \beta)$ for $\bm u \in \RM(r, m-1)/\RM(r-1, m-1) := B, \bm \beta \in \RM(r-1, m-1)/\{1\} \equiv E$. Equation \eqref{pzqrmproofstep} is rewritten as
    \begin{align}
        N_1\sum_{\mathclap{\bm c \in \pRM(r, m)/\{\bm 1\}}}\ket{\bm w \oplus \bm c} 
        &= N_2\sum_{\bm u\in B}\sum_{\bm \beta \in E}\ket{\bm w_1\oplus\bm u^* \oplus \bm \beta^*}\ket{\bm w_2\oplus\bm u \oplus \bm \beta}_{r-1, m-1}^{(0)}\label{pzqrm:decomp:step:1}\\
        &= N_2\sum_{\bm u\in B}\sum_{\bm \beta \in E}\ket{\bm w_1\oplus\bm u^* \oplus \bm \beta^*}\ket{\bm w_2\oplus\bm u}_{r-1, m-1}^{(0)}\label{pzqrm:decomp:step:2}\\
        &= N\sum_{\bm u\in B}\ket{\bm w_1\oplus\bm u^*}_{r-1, m-1}^{(0)}\ket{\bm w_2\oplus\bm u}_{r-1, m-1}^{(0)}\label{pzqrm:decomp:step:3}
    \end{align}
    where $N = 1\Big/ \sqrt{B}$. Equation \eqref{pzqrm:decomp:step:2} is obtained by Properties \eqref{CSS:addx} and \eqref{CSS:Xinv} of CSS codewords as,
    \begin{equation}
        X(\bm \beta) \ket{\bm w}_{r, m}^{(0)} = \ket{\bm w}_{r, m}^{(0)}, ~~ \text{for } \bm \beta \in \RM(r, m)
    \end{equation}
    and Eq. \eqref{pzqrm:decomp:step:3} is obtained using the definition of the codeword. The bounds on the range of $r, r_{\diamond}$ follow from the CSS code requirements and since $\pRM(0, m)/\{1\} = \phi$, we require $r \geq 1$ to ensure the simplification to Eq. \eqref{pzqrm:decomp:step:2} holds. 
\end{proof}
\subsubsection{Constructing recursive encoders}
We denote the encoding circuit for $\pzQRM(r, m; r_{\diamond}, m_{\diamond})$ as $U^{(*, 0)}(r, m; r_{\diamond}, m_{\diamond})$. Similar to the zero rate code, we construct the encoders recursively using $U^{(*, 0)}(r-1, m-1; r_{\diamond}, m_{\diamond})$ and $U^{(0)}(r-1, m-1; r_{\diamond}, m_{\diamond})$ on the first and second sets of qubits respectively. Denote this algorithm as \texttt{RecursiveZPQRM(r, m; r$\_{\diamond}$, m$\_{\diamond}$)} provided in Figure \ref{alg:rec_punc_zero_qrm}.

\begin{figure}
    \begin{mdframed}
    \centering{\noindent\textbf{\texttt{RecursivePZQRM(r, m; r$_{\diamond}$, m$_{\diamond}$)}}, $0\leq r < r_{\diamond}, m < m_{\diamond}$\\
    Returns: Encoder $U^{(0)}(r, m; r_{\diamond}, m_{\diamond})$, qubit index map $M$.}
    \\\underline{Case 1:} If $2r_{\diamond} + 1 \leq m_{\diamond}$ and $r = 0$: \\Return \texttt{RecursiveBasisPQRM(r$_{\diamond}$, m)}
    \\\underline{Case 2:} If $2r_{\diamond} + 1 > m_{\diamond}$ and $m = r_{\diamond}
    + 1$:
    \begin{description}
        \item[1.] Intialize circuit over $2^m-1$ qubits.
        \item[2.] Obtain $U^*(r_{\diamond}, r_{\diamond}+1)$, $M$ from \texttt{RecursivePQRM(r$_{\diamond}$, r$_{\diamond}$+1)}.
        \item[3.] Add gates $H(K)$ where
        $$K = \{i \vert i = M[\bm g] \forall \bm g \in G(2r_{\diamond} + 1 - m_{\diamond}, r_{\diamond}+1)^*\}$$
        \item[4.] Add $U^*(r_{\diamond}, r_{\diamond}+1)$ to the circuit.
        \item[5.] Return circuit, $M$.
    \end{description}
    \underline{Case 3:} Otherwise:\\
    \begin{description}[itemsep=0.5pt]
        \item[1.] Obtain $U^{(*, 0)}(r-1, m-1; r_{\diamond}, m_{\diamond})$, $M_1$ from \texttt{RecursivePZQRM(r-1, m-1; r$_{\diamond}$, m$_{\diamond}$)} and $U^{(0)}(r-1, m-1; r_{\diamond}, m_{\diamond})$, $M_2$ from \texttt{RecursiveZQRM(r-1, m-1; r$_{\diamond}$, m$_{\diamond}$)}.
        
        \item[2.] Initialize circuit over $2^m-1$ qubits. Create $M = \rim(G)$ for 
         $$G = \left[\begin{array}{c} G_1^* \\ \hline G_2^* \vert G_2\end{array}\right], ~~~\text{for } \begin{array}{c} G_1 = G(r_{\diamond}, m)\setminus G(r, m)\cup \{\bm 1\} \\ G_2 = G(r, m-1)\setminus G(r-1, m-1)\end{array}$$
        \item[3.] Add gates $H(K), K = (k+1, \dots, k + {m-1\choose r})$ for $k = 1 + \sum_{i=r+1}^{r_{\diamond}}{m\choose i}$.
        \item[4.] Add qubit permutation $P(p)$ such that 
        $$p(M[(\bm u^*, \bm u)]) = M_1[\bm u^*] ~~\forall \bm u \in G_3,$$ 
        $$p(M[(\bm 0, \bm v)]) = 2^{m-1}-1 + M_2[\bm v] ~~\forall \bm v\in G_4$$
        where $G_3 = G(r_{\diamond}, m-1)\setminus G(r-1, m-1)$ and $G_4 = G(r_{\diamond}-1, m-1)\setminus G(r-1, m-1)$.
        \item[5.] Add $\CNOT{i}{j}$ for $i = M_1[\bm u^*]$, $j = 2^{m-1} + M_2[\bm u]$ $\forall \bm u \in G_3$.
        \item[6.] Add $U^{(*, 0)}(r-1, m-1; r_{\diamond}, m_{\diamond})\otimes U^{(0)}(r-1, m-1; r_{\diamond}, m_{\diamond})$ to the circuit.
        \item[7.] Return circuit, $M$.
    \end{description}
    \end{mdframed}
    \caption{Recursive construction for $U^{(*, 0)}(r, m; r_{\diamond}, m_{\diamond})$}
    \label{alg:rec_punc_zero_qrm}
\end{figure}
On the $i$th  iteration of applying Theorem \ref{thm:decomppz} to decompose $\ket{\bm w}_{r_{\diamond}, m_{\diamond}}^{(*, 0)}$ yields codewords of the form $\ket{\bm w'}_{r_{\diamond}-i, m_{\diamond}-i}^{(*, 0)}$ for $\bm w' \in \pRM(r_{\diamond}, m_{\diamond}-i)$. For this to be a valid punctured zero rate QRM codeword according to the definition in Eq. \eqref{def:pzQRM}, we require 
\begin{align}
    r_{\diamond} - i \geq& 0\\
    m_{\diamond} - i \geq& r_{\diamond} + 1.
\end{align}
Depending on $(r_{\diamond}, m_{\diamond})$, we have two parameter ranges such that there exists non-negative $i$ that satisfies the above two conditions:

\textbf{When $2r_{\diamond} + 1 > m_{\diamond}$}: After $i = m_{\diamond} - r_{\diamond} - 1$ iterations of decomposition, the codeword has the form
\begin{equation}
    \ket{\bm w}_{2r_{\diamond} - m_{\diamond} + 1, r_{\diamond} + 1}^{(*, 0)} \equiv N_1\sum_{\mathclap{\bm c \in D}} \ket{\bm w \oplus \bm c}, ~~~\bm w\in \pRM(r_{\diamond}, r_{\diamond} + 1)
\end{equation}
for $D = \pRM(2r_{\diamond} + 1 - m_{\diamond}, r_{\diamond}+1)/\{\bm 1\}$ and normalization constant, $N = 1/\sqrt{|D|}$ Theorem \ref{thm:decomppz} cannot be applied any further. The encoder is constructed using $\sum_{j = 1}^{r_{\diamond} - i}{m_{\diamond}-i\choose j}$ Hadamard gates and the encoder $U^*(r_{\diamond}, r_{\diamond} + 1)$.

\textbf{When $2r_{\diamond} + 1 \leq m_{\diamond}$}: On $i = r_{\diamond}$ iterations of decomposition, we obtain the classical codeword state 
\begin{equation}
\ket{\bm w}_{0, m_{\diamond} - r_{\diamond}}^{(*, 0)}, ~~~ \bm w \in \pRM(r_{\diamond}, m_{\diamond} - r_{\diamond}).
\end{equation}
This state is then encoded with $U^{(*, c)}(r_{\diamond}, m_{\diamond} - r_{\diamond})$.

\textit{CNOT gate counts}: From the above discussion, the number of CNOT gates are given by
\begin{align}
    \zeta^{(*, 0)}&(r, m; r_{\diamond}, m_{\diamond}) =\nonumber\\& \begin{cases}
        \zeta^{(*, c)}(r_{\diamond}, m) & 2r_{\diamond} + 1 \leq m_{\diamond} \text{ and } r = 0\\
        \zeta^*(r_{\diamond}, r_{\diamond} + 1) & 2r_{\diamond} + 1 > m_{\diamond} \text{ and } r = 2r_{\diamond} - m_{\diamond} + 1\\
        \begin{array}{cc}\sum_{i=r}^{r_{\diamond}}{m-1\choose i}&+ \zeta^{(*, 0)}(r-1, m-1; r_{\diamond}, m_{\diamond})\\ &+ \zeta^{(0)}(r-1, m-1; r_{\diamond}, m_{\diamond})\end{array} & \text{otherwise}
    \end{cases}\label{counts_pzQRM}
\end{align}

\textbf{State preparation circuit $U^{(*, 0, s)}$}: For fault tolerant quantum computing, we require preparation of encoded ancillary states $\ket{0}_{r, m}^{(*, 0)}$ and $\ket{+}_{r, m}^{(*, 0)}$. To construct $\ket{\bm 0}_{r, m}^{(*, 0)}$, in the construction of the encoder, we do not require to add the $\bar 1$ to the generator matrix $G$. The circuits are then modified to reduce gate counts by adding the state-preparation punctured QRM code encoders $U^{(*, s)}(r_{\diamond}, r_{\diamond}+1)$ and $U^{(*, s)}(r_{\diamond}, m_{\diamond}-r_{\diamond})$ at the final iteration instead of $U^{*}(r_{\diamond}, r_{\diamond}+1)$ and $U^{(*, c)}(r_{\diamond}, r_{\diamond} + 1)$ for the parameter ranges $2r_{\diamond} + 1 > m_{\diamond}$ and $2r_{\diamond} + 1 \leq m_{\diamond}$ respectively. The CNOT gate counts in Eq. \eqref{counts_pzQRM} is modified as
\begin{align}
    \zeta^{(*, 0, s)}&(r, m; r_{\diamond}, m_{\diamond}) =\nonumber\\
    &\begin{cases}
        \zeta^{(*, s)}(r_{\diamond}, m) & 2r_{\diamond} + 1 \leq m_{\diamond} \text{ and } r = 0\\
        \zeta^{(*, s)}(r_{\diamond}, r_{\diamond} + 1) & 2r_{\diamond} + 1 > m_{\diamond} \text{ and } r = 2r_{\diamond} - m_{\diamond} + 1 \\
        \begin{array}{cc}\sum_{i=r}^{r_{\diamond}}{m-1\choose i}&+ \zeta^{(*, 0, s)}(r-1, m-1; r_{\diamond}, m_{\diamond})\\ &+ \zeta^{(0)}(r-1, m-1; r_{\diamond}, m_{\diamond})\end{array} & \text{otherwise}
    \end{cases}\label{counts_spzQRM}
\end{align}
The gate counts for a few representative parameters $(r_{\diamond}, m_{\diamond})$ to prepare $\ket{\bm 0}_{r_{\diamond}, m_{\diamond}}^{(*, 0)}$ are provided in Table \ref{tbl:punctured}.
\end{document}